\documentclass[sigconf]{acmart} 
\usepackage{balance}
\usepackage{graphicx,graphics}
\usepackage{amsmath,units}
\usepackage{bm}
\usepackage{tcolorbox}
\usepackage{pifont}
\usepackage{multirow}
\usepackage{color}
\usepackage[noend,linesnumbered,algoruled,boxed,lined]{algorithm2e}
\usepackage{subcaption}
\usepackage{enumitem} 
\usepackage{makecell}
\usepackage{booktabs}
\usepackage{tikz}
\usepackage{circledtext}
\usepackage{natbib}
\usepackage{adjustbox}
\usepackage[skip=0pt]{caption}

\newcommand{\ourmethod}{\textsf{TabularMark}\xspace}

\newcommand{\partitle}[1]{\medskip \noindent \textbf{#1.}}
\newcommand{\nop}[1]{}
\usepackage{xcolor}
\usepackage{color, soul}

\newcommand{\blueno}[1]{{{\textcolor{blue}{#1}}}}
\AtBeginDocument{%
  \providecommand\BibTeX{{%
    \normalfont B\kern-0.5em{\scshape i\kern-0.25em b}\kern-0.8em\TeX}}}

\setcopyright{acmcopyright}
\copyrightyear{2018}
\acmYear{2018}
\acmDOI{XXXXXXX.XXXXXXX}

\acmConference[Conference acronym 'XX]{Make sure to enter the correct
  conference title from your rights confirmation email}{June 03--05,
  2018}{Woodstock, NY}

\acmISBN{978-1-4503-XXXX-X/18/06}

\begin{document}

\title[Watermarking Tabular Datasets for Machine Learning]{\ourmethod: Watermarking Tabular Datasets for \\Machine Learning}


\author{Yihao Zheng}
\affiliation{
    \institution{Zhejiang University}
    \country{}
}
\email{zhengyihao@zju.edu.cn}

\author{Haocheng Xia}
\affiliation{
    \institution{University of Illinois Urbana-Champaign}
    \country{}
}
\email{hxia7@illinois.edu}

\author{Junyuan Pang}
\affiliation{
    \institution{Zhejiang University}
    \country{}
}
\email{junyuanpang@zju.edu.cn}

\author{Jinfei Liu}
\affiliation{
    \institution{Zhejiang University}
    \country{}
}
\email{jinfeiliu@zju.edu.cn}

\author{Kui Ren}
\affiliation{
    \institution{Zhejiang University}
    \country{}
}
\email{kuiren@zju.edu.cn}

\author{Lingyang Chu}
\affiliation{
    \institution{McMaster University}
    \country{}
}
\email{chul9@mcmaster.ca}

\author{Yang Cao}
\affiliation{
    \institution{Tokyo Institute of Technology}
    \country{}
}
\email{cao@c.titech.ac.jp}

\author{Li Xiong}
\affiliation{
    \institution{Emory University}
    \country{}
}
\email{lxiong@emory.edu}

\renewcommand{\shortauthors}{Trovato and Tobin, et al.}

\vspace{+5em}
\begin{abstract}
Watermarking is broadly utilized to protect ownership of shared data while preserving data utility.
However, existing watermarking methods for tabular datasets fall short on the desired properties (detectability, non-intrusiveness, and robustness) and only preserve data utility from the perspective of data statistics, ignoring the performance of downstream ML models trained on the datasets.
Can we watermark tabular datasets without significantly compromising their utility for training ML models while preventing attackers from training usable ML models on attacked datasets?

In this paper, we propose a hypothesis testing-based watermarking scheme, \texttt{\ourmethod}. Data noise partitioning is utilized for data perturbation during embedding, which is adaptable for numerical and categorical attributes while preserving the data utility. For detection, a custom-threshold one proportion z-test is employed \nop{to determine the presence of the watermark}, which can reliably determine the presence of the watermark. Experiments on real-world and synthetic datasets
 {demonstrate the superiority of \ourmethod in detectability, non-intrusiveness, and robustness.
 } 

\end{abstract}


\maketitle

\section{Introduction}\label{sec:intro}




The proliferation of machine learning (ML) has brought significant benefits to a wide range of industries, such as healthcare~\cite{DBLP:journals/cbm/RasheedQGARQ22, jiang2017artificial, wiens2018machine}, retail~\cite{DBLP:journals/corr/abs-2008-07779, chen2017sales, narayana2021machine}, and finance~\cite{DBLP:journals/mansci/HuangS23, dixon2020machine, ozbayoglu2020deep}.  Structured tabular datasets (or relational data) are prevalent and used across many sectors. However, the nature of datasets allows for near-zero cost replication~\cite{rhind1992data}, making them susceptible to unauthorized copying and use. There is a pressing need for a robust mechanism to assert and protect the ownership of such datasets. 

Watermarking is a widely adopted technique for asserting ownership and preventing unauthorized usage of shared data and have been widely employed on multimedia data such as images~\cite{zhong2020automated, ahmadi2021intelligent,DBLP:journals/corr/abs-2311-13713} and audios~\cite{zhang2022robust, yamni2022efficient}, and relational data~\cite{DBLP:journals/corr/abs-1801-08271}. A watermark is embedded into original data usually through subtle perturbations. Data owners can extract the watermark from suspicious data to claim ownership. 
Typically, there are several minimum desired properties of watermarking. 1) detectability: it can be reliably detected, with the aid of some secret information; 2) non-intrusiveness: it should not alter the data in a way that degrades its quality or usability; 3) robustness: it should be resilient to manipulations.  

 Tabular data poses unique challenges due to its data characteristics considering the desired properties above. 1) Tabular data typically consists of precise values with each entry carrying significant and specific information. There is little to no perceptual redundancy compared to multimedia data, which makes it less flexible for designing a watermark that satisfies both detectability and non-intrusiveness, since even minor changes may significantly impact data integrity or usability. 2) The mixture of different data types, including categorical and numerical, may require different or more complex watermarking strategies. 3) Tabular data can undergo a variety of data manipulations like insertions, deletions, and foreign key replacements, and a watermark must be resilient to such operations without being easily removed. 

Watermarking schemes~\cite{zhong2020automated, ahmadi2021intelligent,DBLP:journals/corr/abs-2311-13713,zhang2022robust, yamni2022efficient} that mainly focus on multimedia are well-studied in the literature but are difficult to port over for tabular datasets, due to the dependency of intrinsic patterns or semantic information for a specific multimedia type. In order to extend watermarking outside multimedia, many watermarking schemes~\cite{agrawal2002watermarking, hwang2020reversible, li2022secure, shehab2007watermarking, DBLP:journals/tkde/SionAP04, kamran2013robust, hu2018new, lin2021lrw} have been proposed specifically for relational data. However, they are either not generally applicable to all types of tabular data, or fall short in one or more of the above properties. \cite{agrawal2002watermarking,hwang2020reversible} embed watermarks by modifying the least significant bits (LSBs), causing them to be inapplicable to categorical attributes. For example, embedding bits into encoded categorical attributes (e.g., 0-6) can cause significant distortion and may exceed the value range (e.g., modifying 0110 to 0111), violating \textit{non-intrusiveness}. Schemes proposed in ~\cite{DBLP:journals/tkde/SionAP04,shehab2007watermarking} opt to embed watermark bits into the statistics of the data. However, they necessitate the use of primary keys in the partitioning algorithms. Primary keys are not essential in tabular datasets, and replacing the original primary key with a new column will prevent the correct extraction of the watermarking information, violating \textit{robustness}. In recent years, studies in~\cite{hu2018new,li2020robust} focus on reversible watermarking schemes based
on histogram shifting. These methods are limited to integer attributes and unsuitable for floating-point attributes~\cite{misc_wine_quality_186,misc_abalone_1,misc_auto_mpg_9}.
They also lack a theoretical guarantee on the false positive rate for watermark detection, which weakens the \textit{detectability}. Furthermore, nearly all existing methods~\cite{agrawal2002watermarking, hwang2020reversible, li2022secure, shehab2007watermarking, DBLP:journals/tkde/SionAP04, kamran2013robust, lin2021lrw} primarily measure data utility (for non-intrusiveness) on the basic statistics of data, such as mean and variance for query tasks. Given the current prevalent use of tabular data for building ML models, the added watermark should almost not affect the utility of the downstream models, which we refer to as ML utility. 

In this paper, we address the limitations of existing techniques by proposing a simple yet effective hypothesis testing-based watermarking scheme for tabular datasets, \ourmethod. This scheme partitions data noise into two divisions and introduces designed perturbations to specific cells in the embedding phase. In the detection phase, the deviation distribution characteristics in the suspicious datasets are examined by hypothesis testing.

To ensure detectability, we utilize the one-proportion z-test to detect the perturbations introduced in the embedding phase. As different datasets are assumed to be collected independently, the deviation between a non-watermarked dataset and the original dataset
should be entirely random within the predefined range. Therefore, we use the one-proportion z-test to analyze the characteristics of deviation distribution, which allows for reliable detection of watermarks in watermarked data with a statistically improbable rate of false positives. Additionally, we can adjust the z-test threshold to statistically limit the false positive rate.

To ensure non-intrusiveness, we control the distortion on the ML utility in the embedding phase by managing the number of perturbed cells called \textit{key cells}. Due to randomness in deviation, the probability of the data deviation for a given cell falling within one division is $0.5$. Consequently, as the number of key cells increases, the probability of all of them falling within the chosen division exponentially decreases.  Therefore, compared to the total number of cells, only very few key cells are required to embed a robust watermark. Because the deviation does not rely on a specific data type, the embedding and detection of \ourmethod can be applied across different types of attributes. Whether numerical attributes or categorical attributes are perturbed, the ML utility of watermarked datasets remains almost unchanged thanks to the small number of perturbed cells.

To ensure robustness, we require the data owner to keep the relevant information from the watermark embedding phase confidential, such as the locations of key cells. Without knowing the key cells, attackers can only randomly perturb a large number of cells, much more than the number of key cells, in an attempt to reduce the z-score, an indicator of the z-test.  Therefore, the ML utility of attacked datasets may be significantly reduced. Additionally, we employ the most significant bits (MSBs) of multiple attributes to match key cells in detected datasets to avoid primary key replacement attacks.

Experiments on commonly used real-world datasets validate that the watermark can be reliably detected in watermarked datasets. Additionally, the scheme is proven to be non-intrusive to the ML utility of watermarked datasets and robust against malicious attacks, including insertion, deletion, and alteration attacks. For example, the z-score on the Forest Cover Type increases by 18.6 after watermarking, far exceeding the threshold of 1.96. Even if attackers insert or delete up to 80\% of the tuples from the watermarked Forest Cover Type dataset, the watermark can still be detected. Besides, if attackers successfully remove the watermark from the Forest Cover Type dataset by alteration attacks, the XGBoost model trained on the attacked dataset results in an average drop of 0.245  in the $F_1$-score compared to the model trained on the original dataset. However, the watermarked dataset only causes an average decrease of 0.001 in the $F_1$-score, thereby preserving the ML utility of the dataset. 
We briefly summarize our contributions as follows.
\begin{itemize}[leftmargin=*]
   
    \item We propose a simple yet effective hypothesis testing-based watermarking scheme \ourmethod for tabular datasets using domain partition and one proportion z-test. 
    To our best knowledge, this is the first study that utilizes random deviation to almost completely preserve the ML utility of tabular datasets.
    

    \item To enhance robustness, we employ a hypothesis testing method and multi-attribute matching to mitigate dependencies on data types and primary keys, respectively. Furthermore, we model the watermark removal and mathematically prove the hardness for attackers to remove the watermark.

    \item We demonstrate the detectability, non-intrusiveness, and robustness of \ourmethod on popular real-world datasets for regression and classification tasks as well as explore the trade-offs among multiple hyperparameters of the scheme on synthetic datasets.

\end{itemize}


\section{Related Work}\label{sec:relatedWork}
In this section, we review related research on relational database watermarking~\cite{agrawal2002watermarking, hwang2020reversible, li2022secure, shehab2007watermarking, DBLP:journals/tkde/SionAP04, kamran2013robust, hu2018new, lin2021lrw} and discuss the corresponding limitations.
Although there are other types of watermarking schemes, such as image watermarking~\cite{zhong2020automated, DBLP:journals/corr/abs-2311-13713, ahmadi2021intelligent}, audio watermarking~\cite{zhang2022robust, yamni2022efficient}, and neural network watermarking~\cite{DBLP:conf/kdd/QinYCLDC23,DBLP:conf/icc/WangZCH23,styleauditor2024}, we do not discuss these related works in detail due to the significant differences in data structure from tabular datasets. 


\partitle{Non-Reversible Database Watermarking} The first watermarking scheme~\cite{agrawal2002watermarking} for relational databases was proposed in 2002. They embedded watermarks by modifying the LSB (least significant bit) of certain attributes in some tuples and decided whether to embed a `0' or `1' in the LSB based on the parity of a hash value computed from the concatenated primary keys and private keys. Inspired by this work, subsequent studies~\cite{xiao2007second, hamadou2011weight} improved on embedding a single bit by embedding multiple bits within the selected LSBs. However, these methods are less suitable for categorical attributes, as encoded categorical attributes are mostly small-range integers (e.g., 0-6), offering limited watermarking capacity and potentially causing undesirable distortion.
Another type of relational database watermarking method opts to embed watermark bits into the statistics of the data. ~\citet{DBLP:journals/tkde/SionAP04} used marker tuples to partition the tuples into different subsets, embedding the corresponding watermark bits in various subsets by modifying subset-related statistical information and utilizing a majority voting mechanism to enhance robustness. ~\citet{shehab2007watermarking} improved upon ~\cite{DBLP:journals/tkde/SionAP04} by first optimizing the partitioning algorithm, avoiding the use of marker tuples for subset partitioning and instead using hash values based on primary keys and private keys, effectively resisting insertion and deletion attacks. Furthermore, to minimize distortion, they modeled the embedding of watermark bits as an optimization problem and offered a genetic algorithm and a pattern search method to solve this problem, effectively controlling data distortion. However, the method~\cite{shehab2007watermarking} is strictly limited by the high requirements for data distribution, and it is challenging to define the optimization problem for categorical attributes with discrete and fixed value ranges. Furthermore, it necessitates the use of primary keys for the partitioning algorithm.


\partitle{Reversible Database Watermarking} In recent years, many studies have focused on reversible watermarking schemes, where the watermark can be extracted from the watermarked data and the original data completely restored. Watermarking techniques based on histogram shifting are a promising solution. ~\cite{hu2018new} introduced a reversible watermark algorithm based on histogram shifting in groups. This method also utilized a message authentication code (MAC) calculated with a private key and a primary key for grouping, defined a statistical quantity ``prediction error'' for plotting histograms, and shifted the histogram to embed watermark bits in each group. Similar histogram-based methods include~\cite{li2020robust}, which proposed a robust reversible watermarking mechanism based on consecutive columns in histograms. Methods like those in~\cite{hu2018new, li2020robust} are limited to integer attributes and are not suitable for floating-point attributes which are widely present in tabular datasets for regression tasks. Thus, their application scenarios are constrained. Due to the poor robustness of LSB methods, \citet{li2022secure} proposed a reversible watermarking method by abandoning the LSB approach and embedding watermark bits into the decimal digits. However, it still has the same issues as the LSB methods.

In addition to the above limitations, existing relational database watermarking schemes ignore the effect on the performance of downstream ML models trained on the datasets but only consider changes in the mean or standard deviation of certain attributes for traditional tasks such as query answering. 
The target of this paper is to reduce the ML utility cost of data owners by limiting the impact of schemes on models trained on the watermarked dataset, and to increase the ML utility cost of attackers by aggravating degradation in model performance when removing the watermark.

\section{Algorithms}\label{sec:algorithm}
In this section, we present the watermark embedding and detection algorithms in detail.
The process of \ourmethod is illustrated in Figure~\ref{figure:flowchart} and involves two stages: watermark embedding and watermark detection. Table~\ref{tab:notation} summarizes the frequently used notations.


\begin{table}[t]
\centering
\caption{The summary of frequently used notations.}
\label{tab:notation}
\small %
\begin{tabular}{|c|p{6cm}|}
\hline
Notation & Description \\
\hline
$n$ &  Number of cells from marking attribute \\
$n_w$ & Number of key cells from marking attribute \\
$n_g$ & Number of green key cells \\
$p$ & Positive axis boundary of the perturbation range \\
$k$ & Number of unit domains \\
$\alpha$ & Significance level of the test for detecting a watermark \\
$\beta$ & Proportion affected by a certain attack relative to the overall dataset \\
$\gamma$ & Ratio of the lengths of green domains to red domains \\
$D_o$ & Original dataset \\
$D_w$ & Watermarked dataset \\
$D_s$ & Suspicious dataset \\
\hline
\end{tabular}%
\end{table}


\begin{figure}[h]
    \centering
    \includegraphics[width=\linewidth]{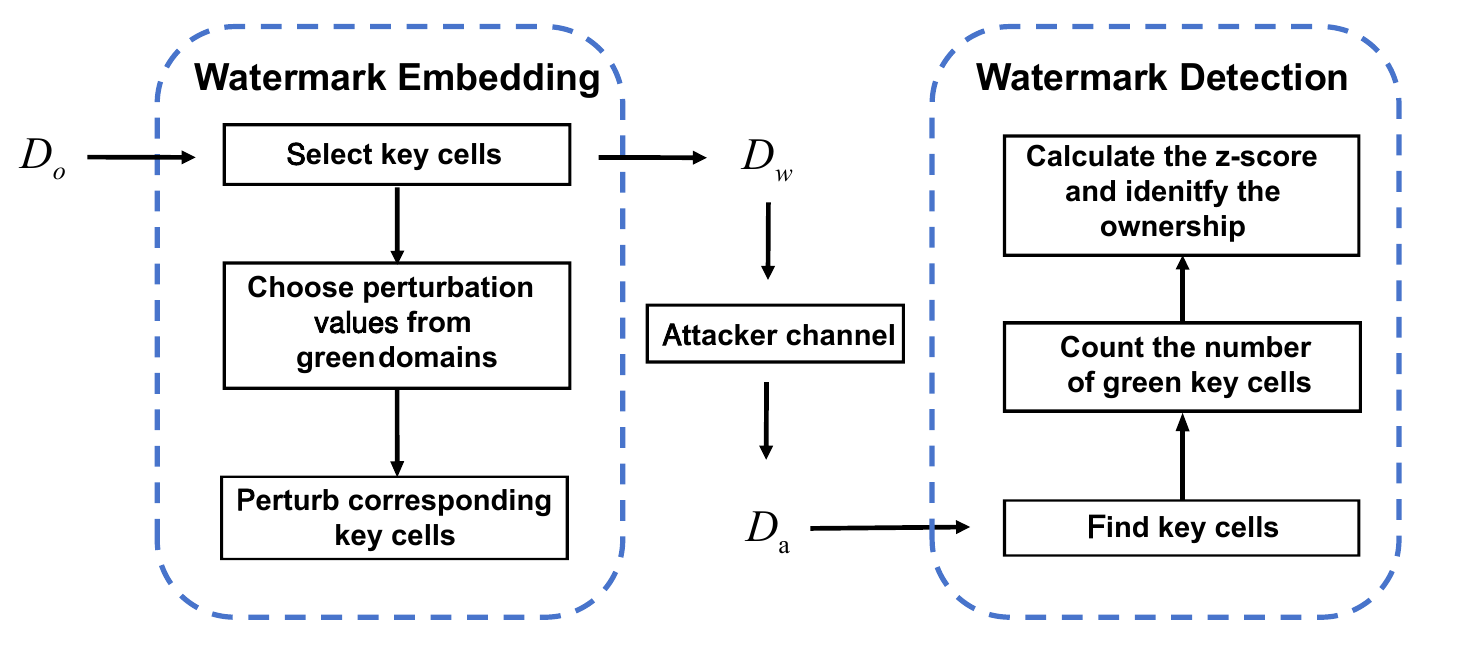}
    \caption{ Flowchart of \ourmethod, where $D_o$ is the original dataset, $D_w$ is the
watermarked dataset, $D_a$ is the watermarked dataset after suffered attacks.}
    \label{figure:flowchart}
\end{figure}


\subsection{Threat Model}

To delimit the scope of this paper, we describe the abilities of attackers as follows.

Considering the characteristics of \ourmethod, we assume that the attackers 1) have full access to the watermarked dataset; 2) are aware of the general process of the watermark embedding and detection algorithms but do not know secret information, such as the distribution of key cells or the specific attributes used for choosing key cells; 3) cannot access the original dataset; and 4) take into consideration not to excessively impact ML utility of the watermarked dataset. For example, in alteration attacks, the range of alteration to the cells would not exceed the perturbation range of the watermark embedding algorithm by a large margin.

\subsection{Watermark Embedding}

In this section, we present the details of the watermark embedding algorithm.

The watermark embedding approach draws inspiration from the random deviations that exist between specific cells within a suspicious dataset and those within the original dataset. This deviation may arise from biases in the measuring instruments used during data collection, or it could stem from differences in the data processing procedures. For example, several institutions independently survey the housing prices in different areas of a city. Due to factors such as the scope of the survey, their data will not be exactly the same and will have certain deviations relative to each other.


If the deviation range is randomly partitioned into two divisions, the probability that the deviation of one cell of suspicious datasets relative to the original dataset falls into one of these divisions is 0.5. If the deviations of multiple cells simultaneously fall into one division, the probability will exponentially decrease. Conversely, the characteristics of the deviation formed by artificially adding perturbation from a particular division can serve as a watermarking scheme. 

It is reasonable to assume that the deviation will not be excessively large. We approximate the deviation range using a custom range \([-p,p]\) and utilize it as the perturbation range. The perturbation range will be partitioned into two divisions, and values will be selected from one division as noise to perturb key cells as a watermark. Figure~\ref{figure:partition} shows an example of a domain partition. We divide the range \([-p, p]\) into \(k\) unit domains which are further divided into green domains and red domains in a 1:1 ratio. As illustrated in Figure~\ref{figure:partition}, the range \([-p, p]\) is partitioned into various red and green domains of different sizes but with an overall equal length. Based on the probability that the deviation of all the detected key cells that fall into the specific type will decrease exponentially, the watermark embedding algorithm only needs to perturb a small number of cells. The characteristics of the deviation distribution can be tested with one proportion z-test, which will be explained in detail in Section~\ref{sec:detection_alg}.

\definecolor{softcoral}{rgb}{0.94, 0.5, 0.5}
\definecolor{lightgreen}{rgb}{0.56, 0.93, 0.56}

\begin{figure}[ht]
    \centering
    \begin{tikzpicture}[scale=0.7]
    \draw[->, thick] (-6,0) -- (6,0) node[right] {};

    \fill[softcoral] (-5,0.25) rectangle (-4.5,-0.25);
    \fill[lightgreen] (-4.5,0.25) rectangle (-2.5,-0.25);
    \fill[softcoral] (-2.5,0.25) rectangle (-1.5,-0.25);
    \fill[lightgreen] (-1.5,0.25) rectangle (-0.5,-0.25);
    \fill[softcoral] (-0.5,0.25) rectangle (0,-0.25);
    \fill[lightgreen] (0,0.25) rectangle (0.5,-0.25);
    \fill[softcoral] (0.5,0.25) rectangle (2.5,-0.25);
    \fill[lightgreen] (2.5,0.25) rectangle (3.5,-0.25);
    \fill[softcoral] (3.5,0.25) rectangle (4,-0.25);
    \fill[lightgreen] (4,0.25) rectangle (5,-0.25);
    
    \foreach \x/\label in {-5/-p,0/0,5/p}
    \draw (\x, 0.35) -- (\x, -0.35) node[below] {$\label$};

    \begin{scope}[yshift=1cm, xshift=-0.2cm]
        \fill[softcoral] (-3.5,0) rectangle (-3,0.5);
        \node[right] at (-3,0.25) {Red domain};
        \fill[lightgreen] (1,0) rectangle (1.5,0.5);
        \node[right] at (1.5,0.25) {Green domain};
    \end{scope}
\end{tikzpicture}
    \caption{An example of domain partition.}
    \label{figure:partition}
\end{figure}
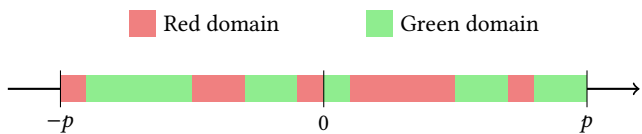


For brevity, this paper focuses on the widely used tasks of classification and regression within the realm of ML. Tabular datasets employed by these tasks chiefly comprise two types of attributes: numerical attributes and categorical attributes, which can be used by \ourmethod to embed watermarks. Besides, we do not consider the character traits of categorical attributes, but only the encoded categorical attributes such as 0-6. Hence, the categorical attributes can be regarded as special numerical attributes.

\begin{figure}[h]
    \vspace{0.5cm}
    \centering
    \includegraphics[width=\linewidth]{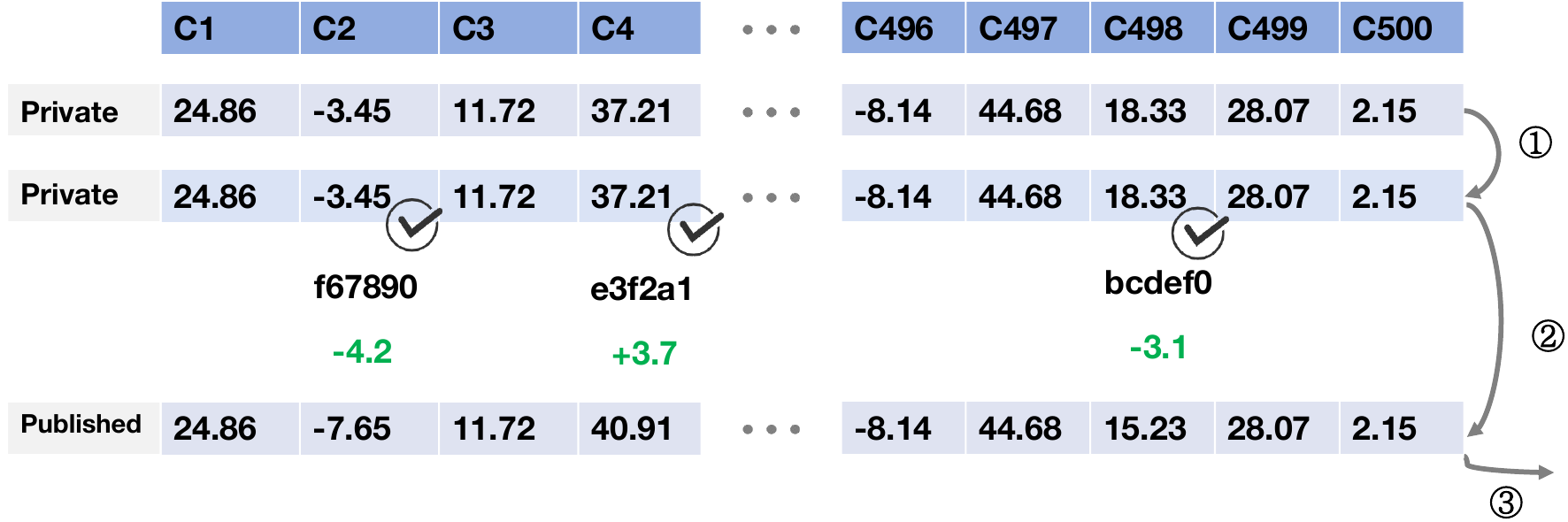}
    \caption{An example of watermark embedding.}
    \label{figure:embedding}
\end{figure}

A tabular dataset contains attributes $A_{0}, A_{1}, \ldots, A_{v-1}$ and we assume ML models using $A_{v-1}$ as the prediction target. For simplicity, we select key cells from one attribute. In Section ~\ref{sec:experiments}, we verify the effectiveness of \ourmethod by embedding the watermark into the attribute $A_{v-1}$. The algorithm can easily extend when the watermarking scheme selects key cells from multiple attributes. Algorithm~\ref{alg:watermark_generation_numerical} gives detailed information on the watermark embedding algorithm when we select key cells from numerical attributes. We first choose an attribute $A_i$ to which the watermark will be embedded (Line 1). Then, the number of key cells n and the boundary of the perturbation range p are determined based on the strength of the watermark (Line 2). On the original dataset $D_o$, we pick n cells as key cells from the chosen attribute $A_i$ (Line 3). For each key cell, we divide the perturbation range $[-p,p]$ into k domains equally, then randomly divided into \(0.5k\) green domains and \(0.5k\) red domains, and the random number seed is kept as a secret. In practice, we tend to choose a relatively large value for \(k\), such that within a small part of \([-p, p]\), the overall lengths of green domains and red domains approach equality. Perturb the key cell by randomly choosing a number from the green domains (Lines 4-9). After perturbing all key cells, we finish the watermark embedding process. The algorithm can be converted to categorical attributes with simple modifications. As the value range of encoded categorical attributes is discrete and relatively small, it does not necessitate a predefined perturbation range. We choose to randomly divide the possible categories into two domains, green and red, and then select a category from the green domain to replace the original category attribute.

\begin{example}
Figure~\ref{figure:embedding} shows an example of the watermark embedding algorithm for numerical attributes. Suppose the attribute we embed a watermark has 500 cells. \ding{172} We select 50 key cells from these cells, and 3 key cells out of them are indicated by checkmarks. We set the perturbation range to \([-5, 5]\). \ding{173} For each key cell, we use a random number to seed a random number generator. The random seeds corresponding to each key cell are represented below in hexadecimal format. Utilizing these seeds, we partition the range \([-5, 5]\) as illustrated in Figure~\ref{figure:partition}. Since each key cell has a unique seed, the resulting green and red domains are also different for each cell. Then, for each key cell, we select a random number as data noise from its corresponding green domains (represented by the green numbers under the random seeds) to perturb it. \ding{174} After perturbing all key cells in this manner, we obtain the watermarked dataset, which is then ready for publishing. 
\label{example:embedding}
\end{example}



\begin{algorithm}[ht] \caption{Watermark Embedding Algorithm.}
\label{alg:watermark_generation_numerical}
\SetKwInOut{Input}{input}\SetKwInOut{Output}{output}

\Input{original dataset $D_o$, selected attribute $A_i$}
\Output{watermarked dataset $D_w$}
    choose an attribute $A_i$ to which the watermark is to be added\;
    determine the values of p and n\;
    select $n_w$ cells as key cells from the selected attribute\;
    \ForEach{key cell}{
    use a random number to seed a random number generator\;
    partition the range $[-p, p]$ into $k$ unit domains\;
    with this seed, randomly divide these domains into 0.5k green domains and 0.5k red domains\;
    from the green domains, select a random number\;
    perturb the key cell with this random number\;
    }

\end{algorithm}




\subsection{Watermark Detection} \label{sec:detection_alg}
In this section, we provide the details of the watermark detection algorithm. Specifically, the watermark detection algorithm employs a rigorous statistical measure to examine suspicious datasets, ensuring successful detection even in the face of substantial attacks. Next, we introduce this measure in detail.

\subsubsection{One Proportion Z-test} The one proportion z-test~\cite{fleiss2013statistical} is the statistical tool utilized in the watermark detection algorithm. One proportion z-test is a statistical test that determines whether a single sample rate (e.g., the success rate) is significantly different from a known or hypothesized population rate. We define the null hypothesis $H_0$ to detect the watermark as follows, 


\begin{tcolorbox}[
colback=white,
colframe=black,
boxsep=-1mm,
boxrule=1pt]
    \begin{equation*}
        \begin{split}
        H_0:\quad & \text{ The tabular dataset is not generated} \\
             & \text{ by attacking the watermarked dataset.}
        \end{split}
    \end{equation*}
\end{tcolorbox}

Then we will test a statistic for a one-proportion z-test as 
\[
z = \frac{2 \times( n_g - 0.5n_w)}{\sqrt{n_w}},
\]
where \( n_w \) is the total number of key cells and \( n_g \) is the number of green cells counted in the detection phase. The meaning of this equation is to subtract the expected value from the observed value, and then divide by the standard deviation of the observed value. In the watermark detection algorithm, $n_g$ is the main observed value. Given that in the watermark embedding algorithm, the length ratio of green domains to red domains is one-to-one, the probability of the deviations of the detected key cells relative to the key cells in the original dataset falling within the green domains can be considered as 0.5. Therefore, the expectation of $n_g$ is \(0.5n_w\) and its standard deviation is \(\sqrt{0.25n_w}\) .

 A significance level \( \alpha \) needs to be defined according to the strength of the scheme. Commonly, \( \alpha = 0.05 \) is used, which implies a 5\% risk of incorrectly rejecting the null hypothesis ($H_0$). Subsequently, we need to calculate the p-value. The p-value is the probability of observing a test statistic as extreme as, or more extreme than, the value observed, under the assumption that the null hypothesis is true. If the p-value is less than or equal to \( \alpha \), we reject the null hypothesis. This means that there is statistical evidence to suggest that the sample proportion is significantly different from the expected proportion. Conversely, if the p-value is greater, we fail to reject the null hypothesis, suggesting that any observed difference could reasonably occur by random chance, implying that no watermark is detected in the suspicious dataset $D_s$. The \( p \)-value can directly measure the magnitude of the false positive rate because a larger \( p \)-value indicates a higher probability of observing the sample under the null hypothesis condition, which means a higher false positive rate when we reject the null hypothesis.

\subsubsection{Detection} In the watermark detection algorithm, \( \alpha \) will be set directly to the value related to the z-score. As z-score and p-value are negatively correlated, the larger the z-score, the smaller the p-value. Users can customize the threshold to adjust the strength of the watermark detection algorithm. If the threshold of the p-value is set to 0.05, the corresponding z-score threshold is 1.96. In Section~\ref{sec:experiments}, we default the z-score threshold to 1.96.

Algorithm~\ref{alg:watermark_detection_numerical} details the steps of the watermark detection algorithm for numerical attributes. For each key cell of $D_s$, the partition of green domains and red domains is recovered using the corresponding random seed (Lines 3-4). Calculate the difference of this key cell between $D_s$ and $D_o$, and determine whether the difference belongs to the divided green domains. If so, we mark this key cell as a green cell (Lines 5-7). The process concludes with the calculation of the z-score, which is compared to a predetermined threshold \(\alpha\). If the computed z-score surpasses \(\alpha\), it indicates the presence of the watermark in the dataset \(D_s\). In contrast, if the z-score falls below \(\alpha\), it suggests that the watermark is not detected in \(D_s\) (Lines 8-12). The algorithm can be converted to categorical attributes by directly determining whether a key cell in \( D_s \) belongs to the green domain to count the green cells.

\begin{figure}[t] 
\centering 
\includegraphics[width=\linewidth]{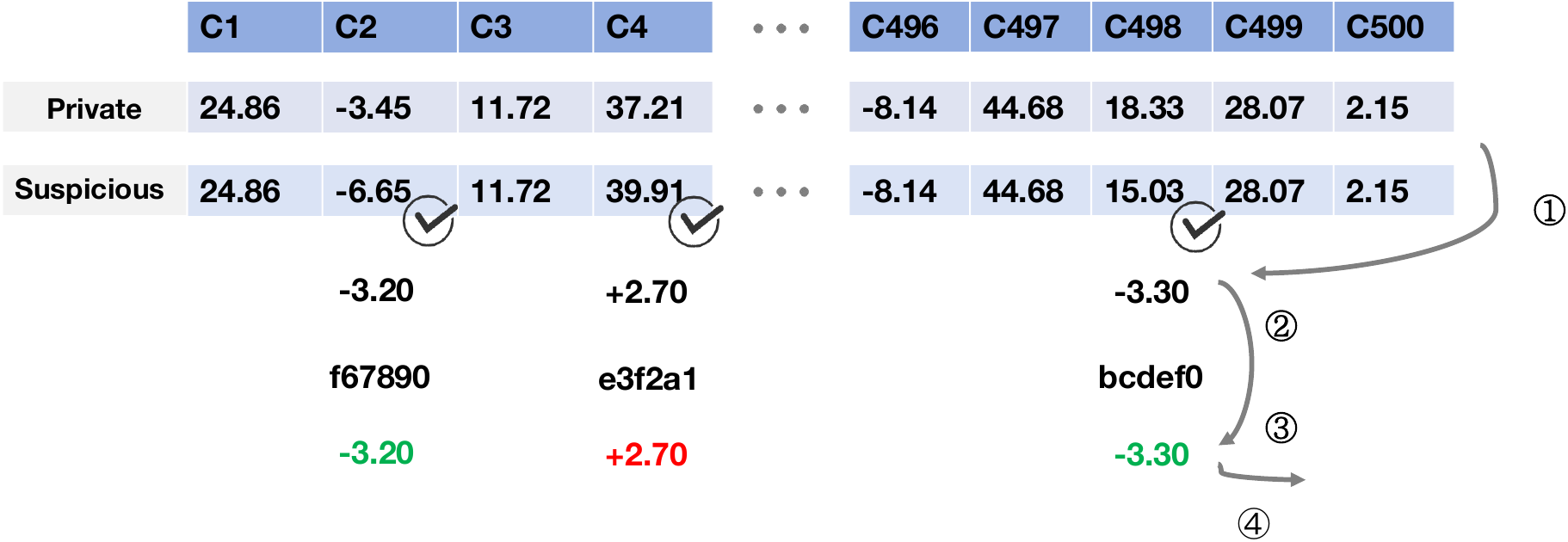} \caption{An example for watermark detection.} 
\label{fig:watermark_detection_example}
\end{figure}


\begin{example}\label{example:detection}
An example of the watermark detection algorithm is shown in Figure~\ref{fig:watermark_detection_example}. \ding{172} First, we calculate the difference of 
all 50 key cells between the private data and the suspicious data. \ding{173} Next, we use the random seeds previously employed in watermark embedding to restore the green domains and red domains corresponding to each key cell. \ding{174} If the difference of a key cell in the suspicious data falls within its corresponding green domains, the key cell is marked as a green cell. Otherwise, it is marked as a red cell. We find that 30 out of the total 50 key cells are green, this indicates that 40\% of the green cells have been flipped by the attacker into red cells. \ding{175} Given that the total number of key cells is 50, the corresponding z-score is calculated to be 1.414. By determining whether the z-score is greater than $\alpha$, we can determine whether the suspicious data is watermarked.
\end{example}

\begin{algorithm}[t] \caption{Watermark Detection Algorithm.}
\label{alg:watermark_detection_numerical}
\SetKwInOut{Input}{input}\SetKwInOut{Output}{output}

\Input{original dataset $D_o$, suspicious dataset $D_s$, selected attribute $A_i$, perturbation range $[-p, p]$, number of key cells $n_w$}
\Output{\rm{true} or \rm{false}}
$n_g \gets$ 0\;
\ForEach{key cell}{
use the corresponding random number to seed a random number generator\;
restore the green domains and red domains of $[-p, p]$\;
calculate the difference between a key cell of $D_s$ and that of $D_o$\;
 \If{the difference $\in$ green domains}{
        $n_g \gets n_g + 1$\;
    }
}
$z$ $\gets$ $\frac{2( n_g - 0.5n_w)}{\sqrt{n_w}}$\;
\eIf{z $\ge$ $\alpha$}{
        \Return{\rm{true}};
    }{
        \Return{\rm{false}};
    }    
\end{algorithm}



\subsubsection{Matching} In many scenarios, the relative positions of the tuples requiring examination in the detected dataset are inconsistent with those in the watermarked dataset. This inconsistency can be caused by attacks, such as insertion or deletion attacks, which may change the relative positions of tuples. Therefore, it is necessary to locate these perturbed tuples. Considering that watermark schemes reliant on primary keys are susceptible to primary key replacement attacks, we need a matching algorithm to address this issue. Inspired by the basic fact that for a tabular dataset, the probability that the MSBs (most significant bits) of multiple attributes are simultaneously equal at the same time is low, and it's difficult for attackers to substitute them all, we choose MSBs of several attributes to form a special primary key and check whether the primary key of a tuple in $D_s$ is equal to the primary key of a tuple containing a key cell in $D_o$. 

Algorithm~\ref{alg:matching} details the steps to locate tuples with key cells in $D_s$. We first select MSBs of k attributes as a special primary key (Line 1). For each tuple containing key cells in $D_o$, we use the corresponding primary key to compare with the primary key of each tuple in $D_d$. We find the corresponding key cell when the tuple in $D_d$ is matched (Lines 3-8). Our practice shows that the effect of the watermark detection algorithm is already effective when two attributes are chosen, and the selected attributes are almost equivalent to a true primary key when we choose three attributes. 

\begin{algorithm}[t] \caption{Matching Algorithm.}
\label{alg:matching}
\SetKwInOut{Input}{input}\SetKwInOut{Output}{output}

\Input{original dataset $D_o$, detected dataset $D_d$}
\Output{key cells in $D_o$}
select MSBs of k attributes as a primary key\; 
$k_ds \gets$  \( \emptyset \) \;
\ForEach{tuple containing key cell in $D_d$}{
\ForEach{tuple in $D_o$}{
    \If{primary key of $t_o$ = primary key of $t_d$}{
        $k_d \gets$ key cell in $t_d$\;
        \( k_ds \gets k_ds \cup \{k_d\} \)\;
    }
}
\Return{$k_ds$};
}

\end{algorithm}



To combine Algorithm~\ref{alg:watermark_detection_numerical} with Algorithm~\ref{alg:matching}, we simply run Algorithm~\ref{alg:matching} at the beginning of Algorithm~\ref{alg:watermark_detection_numerical} to match tuples, which is designed to identify the perturbed tuples containing key cells, especially in scenarios where the order of tuples is disrupted due to some attacks such as insertion or deletion attacks.

\begin{figure}[t]
    \centering
    \includegraphics[width=\linewidth]{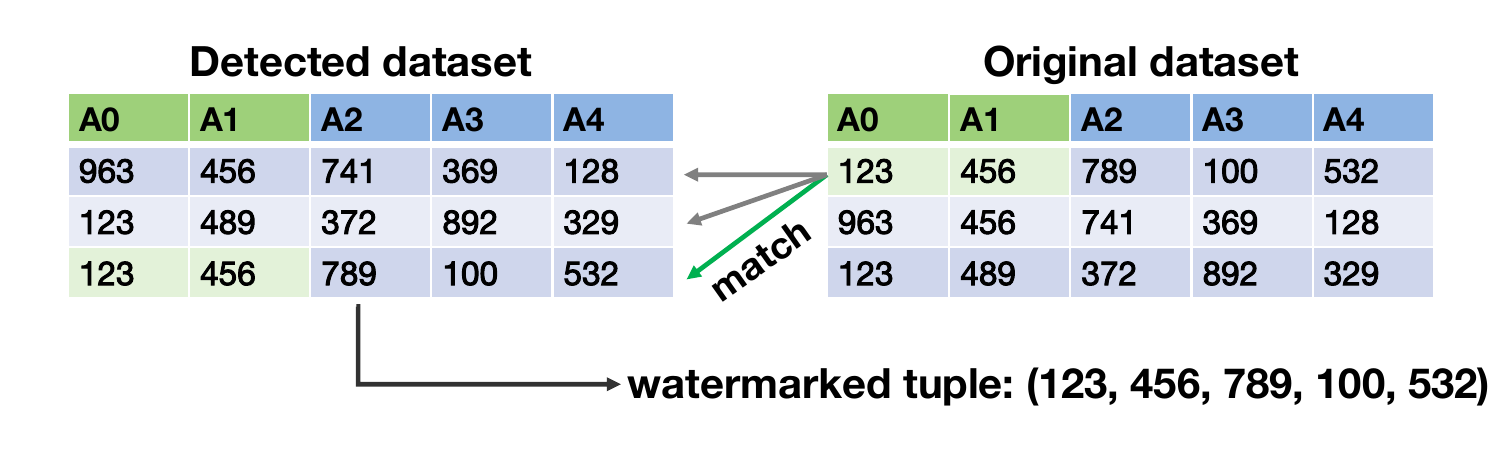}
    \caption{An example of matching tuples.}
    \label{figure:matching}
\end{figure}

\begin{example}
Figure~\ref{figure:matching} shows an example of using Algorithm~\ref{alg:matching} to find perturbed tuples. The left side of the figure shows the detected dataset and the right side is the original dataset. It can be observed that the order of the detected dataset is scrambled compared to the original dataset, making it difficult to find perturbed tuples directly. We assume that the first tuple of the original dataset is the watermarked tuple containing a key cell. We then select MSBs of the first two attributes, $A_0$ and $A_1$, as a special primary key. Using this special primary key, we compare tuples in the detected dataset. When we compare to the third tuple of the detected dataset, both values of MSBs of $A_0$ and $A_1$ equal the values in the special primary key of the tuple. Finally, we successfully identify the perturbed tuple in the detected dataset.
\label{exampe:matching}
\end{example}
\subsection{Analysis on Watermark Removal}
In this section, to demonstrate the robustness of \ourmethod, we model the potential strategies of alteration attacks and mathematically prove the lower bound of the number of cells that need to be altered by the attacker to achieve a high confidence level. 
According to our partition strategy, when a green cell is perturbed out of the range of the unit green domain, there is a $0.5$ chance that it will be flipped to a red cell.
According to the threat model, the attackers cannot discern the key cells. Therefore, to flip green key cells into red, attackers can only randomly select \(n_h\) cells and introduce noise. We simulate alteration attacks using three representative methods of adding noise, including uniform distribution noise, Gaussian distribution noise, and Laplace distribution noise, all with a mean of 0 and a standard deviation of \(\sigma\).

\begin{proposition}\label{prop:erasure}
    Denote by $n_h$ the amount of cells an attacker tamper with. When each cell is added an \textnormal{i.i.d} noise from uniform distribution $\epsilon \sim U[-2\sigma,2\sigma]$. The expectation of $n_h$ for the watermark removing is $$\mathrm{E}[n_h]=\frac{n(n_w-n_\alpha)}{n_wp_\sigma},$$
    where $n_\alpha = \alpha \sqrt{\frac{n_w}{4}}+\frac{n_w}{2}$ is the least amount of green cells to make the z-score achieve $\alpha$ and $p_\sigma=0.5-\frac{p}{4k\sigma}$. For example, when $n=10000,n_w=400$ and $n_\alpha = 250$, we have $\mathrm{E}[n_h] > 7500$ since $p_\sigma < 0.5$. Moreover, to achieve a $95\%$ confidence interval, we have
    $$
    n_h > n_w + (\frac{(n-n_w)!}{0.05n!}\sum_{n_x=0}^{n_w-n_\alpha}\sum_{n_y=n_x}^{n_w}\binom{n_y}{n_x}\binom{n_w}{n_y}p_\sigma^{n_x}(1-p_\sigma)^{n_y-n_x})^{\frac {-1}{n_w}}.
    $$
\end{proposition}
\begin{proof}
    We state the case where $\gamma = 0.5$ and $\sigma \geq p/k$. Consider a green cell $x\sim U[0,2p/k]$, the probability $x+\epsilon$ is not in the same interval as $x$ is $Pr(x+\epsilon>2p/k)+Pr(x+\epsilon<0)$, which is the $Pr(\epsilon>2p/k-x)+Pr(\epsilon<-x)=1-\frac{p}{2k\sigma}$. Therefore, the probability that $x+\epsilon$ is red is $p_\sigma = 0.5-\frac{p}{4k\sigma}$.

    To achieve the watermark removing, we need to add noise on at least $n_w-n_\alpha$ cells with the watermark. Therefore, the expectation of $n_h$ is $\frac{n(n_w-n_\alpha)}{n_w}(0.5-\frac{p}{4k\sigma})^{-1}$.

    Consider the case of a $95\%$ confidence interval. Denote by $N_x$ the amount of key cells changing from green to red. Then, $Pr(N_x=n_x)=\sum_{n_y=n_x}^{n_w}\frac{\binom{n_y}{n_x}\binom{n_w}{n_y}\binom{n-n_w}{n_h-n_y}}{\binom{n}{n_h}}p_\sigma^{n_x}(1-p_\sigma)^{n_y-n_x}>\frac{\binom{n-n_w}{n_h-n_w}}{\binom{n}{n_h}}$ $\sum_{n_y=n_x}^{n_w}$ $\binom{n_y}{n_x}\binom{n_w}{n_y}$ $p_\sigma^{n_x}(1-p_\sigma)^{n_y-n_x}$ when we consider $n_h$ is not very large, which means $\sum_{n_x=0}^{n_w-n_\alpha}Pr(N_x=n_x)>(n_h-n_w)^{n_w}$ $\frac{(n-n_w)!}{n!}\sum_{n_x=0}^{n_w-n_\alpha}$ $\sum_{n_y=n_x}^{n_w}\binom{n_y}{n_x}\binom{n_w}{n_y}$ $p_\sigma^{n_x}(1-p_\sigma)^{n_y-n_x}$. Then, we can derive the aforementioned inequality.
\end{proof}

\begin{proposition}
    Denote by $n_h$ the amount of cells an attacker tamper with. When each cell is added an \textnormal{i.i.d} noise from Gaussian distribution $\epsilon \sim \mathcal{N}(0,\sigma^2)$. The expectation of $n_h$ for the watermark removing is $$\mathrm{E}[n_h]\geq\frac{n(n_w-n_\alpha)}{n_wp_\sigma},$$
     where $n_\alpha = \alpha \sqrt{\frac{n_w}{4}}+\frac{n_w}{2}$ is the least amount of green cells to make the z-score achieve $\alpha$, $p_\sigma=\frac 14 + \frac 12 \int_{\frac{2p}{k\sigma}}^{\infty}\psi(x)dx$, and $\psi(\cdot)$ denotes probability density function (PDF) of the standard normal distribution. Moreover, to achieve a $95\%$ confidence interval, we have
    $$
    n_h > n_w + (\frac{(n-n_w)!}{0.05n!}\sum_{n_x=0}^{n_w-n_\alpha}\sum_{n_y=n_x}^{n_w}\binom{n_y}{n_x}\binom{n_w}{n_y}p_\sigma^{n_x}(1-p_\sigma)^{n_y-n_x})^{\frac {-1}{n_w}}.
    $$
\end{proposition}
\begin{proof}
    Similarly, we only need to compute the maximum probability that a green cell is changed to red and substitute it into $p_\sigma$ of Proposition~\ref{prop:erasure}, with $Pr(\epsilon>2p/k-x)+Pr(\epsilon<-x)\leq Pr(\epsilon>2p/k)+Pr(\epsilon<0)=\frac 12 + \int_{\frac{2p}{k\sigma}}^{\infty}\psi(x)dx$.
\end{proof}
\begin{proposition}
    Denote by $n_h$ the amount of cells an attacker tamper with. When each cell is added an \textnormal{i.i.d} noise from Laplace distribution $\epsilon \sim \mathrm{Lap}(0,\sigma/\sqrt{2})$. The expectation of $n_h$ for the watermark removing is $$\mathrm{E}[n_h]\geq\frac{n(n_w-n_\alpha)}{n_wp_\sigma},$$
     where $n_\alpha = \alpha \sqrt{\frac{n_w}{4}}+\frac{n_w}{2}$ is the least amount of green cells to make the z-score achieve $\alpha$ and $p_\sigma=\frac 14 (1+\exp\{-\sqrt{2}p/k\sigma\})$. Moreover, to achieve a $95\%$ confidence interval, we have
    $$
    n_h > n_w + (\frac{(n-n_w)!}{0.05n!}\sum_{n_x=0}^{n_w-n_\alpha}\sum_{n_y=n_x}^{n_w}\binom{n_y}{n_x}\binom{n_w}{n_y}p_\sigma^{n_x}(1-p_\sigma)^{n_y-n_x})^{\frac {-1}{n_w}}.
    $$
\end{proposition}
\begin{proof}
    Similarly, we only need to compute the maximum probability that a green cell is changed to red, with $Pr(\epsilon>2p/k-x)+Pr(\epsilon<-x)\leq Pr(\epsilon>2p/k)+Pr(\epsilon<0)=\frac 12 + \int_{\frac{\sqrt{2}p}{k\sigma}}^{\infty}\frac 12 e^{-x}dx = \frac 12 (1+\exp\{-\sqrt{2}p/k\sigma\})$.
\end{proof}


\partitle{Analysis and Conclusions} \(n_h\) required to erase the watermark with a 95\% confidence is obviously greater than the expected \(n_h\) when erasing the watermark, so we analyze the robustness of the watermark from \(E[n_h]\). The upper bound of $p_\sigma$ is strictly less than $0.5$, since there is still a fifty percent chance that the changed interval remains green. Moreover, the upper bound of $p_\sigma$ is greater in the case of Gaussian noise or Laplace noise when the variance of the noise is small, which reflects the potential of the latter two strategies to reduce $E[n_h]$ in such cases.

For \(E[n_h]\), even if we assume that every perturbation can flip green cells, we have \(E[n_h] \ge \frac{n(n_w-n_\alpha)}{n_wp_\sigma} = \frac{n(n_w-n_\alpha)}{n_w} = n\left(1- \frac{n_\alpha}{n_w}\right)\) when added noises are from uniform distribution, Gaussian distribution noise, or Laplace distribution. Considering that \(n_\alpha\) approaches half of \(n_w\), \(E[n_h]\) is also close to half of \(n\), much greater than \(n_w\), indicating that attackers need to pay a much higher cost than the data owner to erase the watermark.

\section{Experiments and Analysis}\label{sec:experiments}

In this section, we organize experiments by answering five research questions (RQs). Following the three desiderata mentioned in Section~\ref{sec:intro}, we design four related RQs to demonstrate the superiority of \ourmethod. To further explore the usage of \ourmethod, we complement one additional RQ on hyper-parameter selection.

\begin{itemize}[leftmargin=3em]
    \item[\textbf{RQ1)}] 
    Is the \ourmethod detectable for various tabular datasets? (Section \ref{subsec:rq1}) 
    \item[\textbf{RQ2)}] 
    Can \ourmethod achieve non-intrusiveness across various tabular datasets used for ML? (Section \ref{subsec:rq2}) 
    \item[\textbf{RQ3)}] 
    Can \ourmethod be robust against common malicious attacks? (Section \ref{subsec:rq3}) 
    \item[\textbf{RQ4)}]
    How does \ourmethod perform compared to the related state-of-the-art schemes? (Section \ref{subsec:rq4})
    \item[\textbf{RQ5)}] 
    What are the trade-offs among the hyper-parameters of \ourmethod? (Section \ref{subsec:rq5}) 
\end{itemize}

Moreover, we attempt to modify the perturbation approach to optimize \(\text{\ourmethod}\), and in Section~\ref{subsec:rq6}, we validate that the optimization strategy possesses better non-intrusiveness than the original.


\partitle{Experimental Setup}
Experiments are conducted on a server comprising two Intel(R) Xeon(R) Platinum 8383C CPU@2.70GHz\nop{and four Geforce RTX 3090 GPUs}, running Ubuntu 18.04 LTS 64-bit with 256GB memory. We use NumPy to randomly generate two-dimensional normal distribution data as synthetic datasets, with a mean of \(\mu\) = 0, a variance of \(\sigma\) = 20, and a total of 2000 tuples to run. Forest Cover Type~\cite{misc_covertype_31}, HOG feature based on the digits dataset~\cite{misc_pen-based_recognition_of_handwritten_digits_81}, and Boston Housing Prices~\cite{HARRISON197881} are employed as real-world datasets. Notably, the HOG feature dataset is generated with the histogram of oriented gradients (HOG) features extracted from the digits dataset, combined with their categories.
For ML models, XGBoost~\cite{chen2016xgboost}, Random Forest~\cite{breiman2001random}, and Linear Regression~\cite{montgomery2021introduction} are used to cover both classification and regression tasks. 
Besides, the z-score threshold defaults to 1.96 to maintain the theoretical constraint of a 5\% false positive rate.

\partitle{Reproducibility}
For reproducibility, we set a specific random seed to ensure that the data distortion and model training effects are repeatable on the same hardware and software stack. To eliminate the randomness of the results, the experimental outcomes are generally averaged over multiple runs.

\subsection{Detectability}\label{subsec:rq1}
In this section, we investigate the z-scores on the original, watermarked, perturbed datasets to answer \textbf{RQ1}. One synthetic and three real-world datasets are utilized to verify the detectability of \ourmethod. 


In the experiment, synthetic datasets are watermarked with \( n_w = 300 \), \( p = 2\sigma \), and \( k = 500 \) on the first dimension. Forest Cover Type dataset is watermarked with \( n_w = 300 \) on the attribute Cover\_Type. HOG feature dataset is watermarked with \( n_w = 150 \) on the attribute \textit{category}. Boston Housing Prices dataset is watermarked with \( p = 25 \), \( k = 500 \), and \( n_w = 50 \) on the attribute \textit{MEDV}.
Subsequently, watermark detection is performed on the original dataset $D_o$, the watermarked dataset $D_w$, and the dataset $D_p$ generated after randomly perturbing $D_o$. Results of detected z-scores are shown in Table~\ref{tab:falsehit}. 



\begin{table}[h!]
  \centering
  \caption{Z-scores.}
  \label{tab:falsehit}
  \begin{tabular}{@{}p{2.3cm}>{\centering\arraybackslash}p{1.8cm}>{\centering\arraybackslash}p{1.8cm}>{\centering\arraybackslash}p{1.2cm}@{}}
    \toprule
    Dataset & $D_o$ & $D_w$ & $D_p$ \\
    \midrule
    Synthetic & -0.0531 & 17.3 & -0.0393 \\
    Forest & -0.0159 & 18.6  & -0.0467\\
    HOG & -0.327 & 12.3 & 0.128\\
    Boston Housing & -0.113 & 6.91 & -0.142\\
    \bottomrule
  \end{tabular}
\end{table}

The z-scores in \(D_w\)s are all significantly larger than the threshold 1.96, and the z-scores on the other two types of datasets are significantly smaller than the threshold, suggesting that \ourmethod is detectable.
In addition, we utilized $500 \pm 10$ synthetic datasets to plot the ROC charts in Figure ~\ref{fig:ruc}. The curve indicates that \ourmethod can reliably identify the presence of a watermark in the vast majority of cases.

\begin{figure}[ht] 
\centering 
\includegraphics[width=0.8\linewidth]{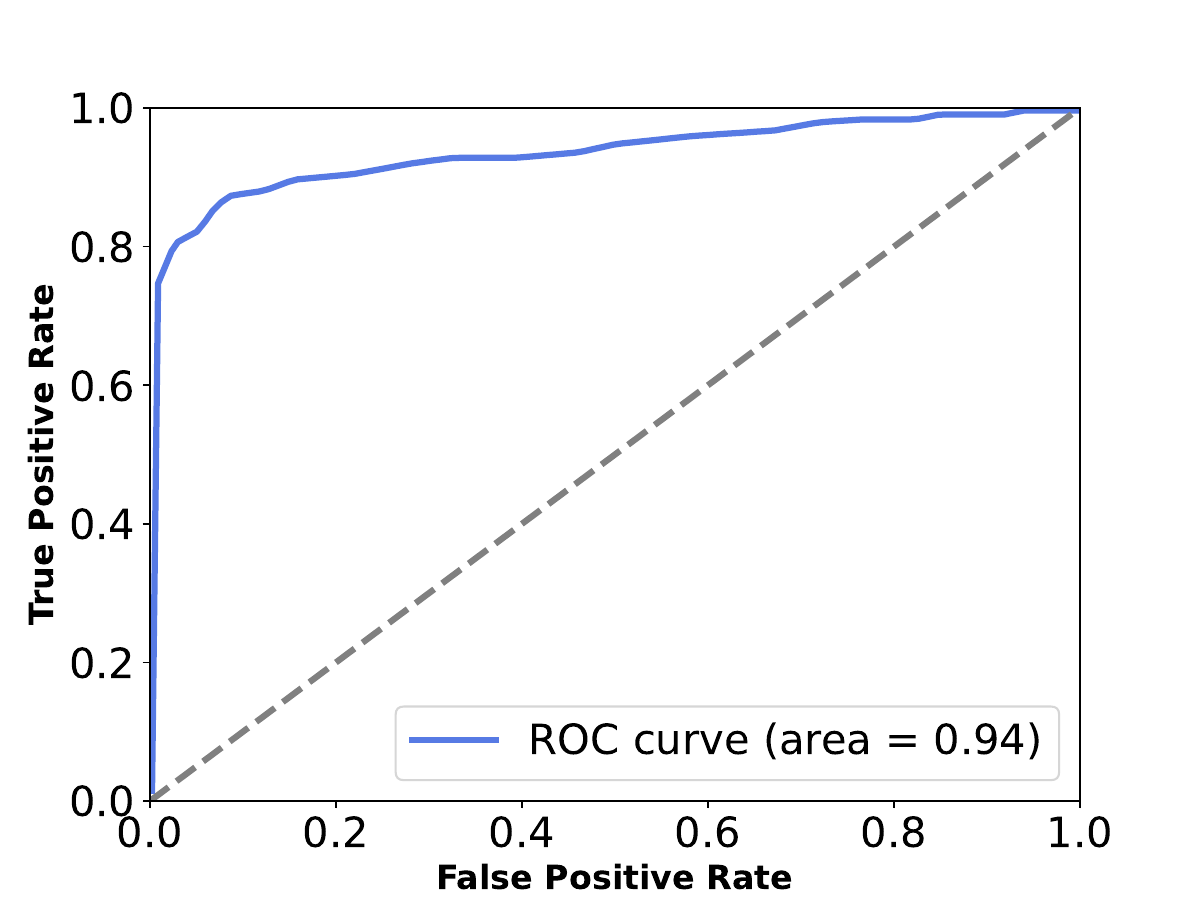} 
\caption{ROC curve.} 
\label{fig:ruc} 
\end{figure}


\subsection{Non-Intrusiveness} 
\label{subsec:rq2}
In this section, we investigate ML-related metrics on the original and watermarked datasets to answer \textbf{RQ2}. Three real-world datasets, including the Forest Cover Type dataset, the HOG feature dataset, and the Boston Housing Prices dataset, are utilized to verify the non-intrusiveness of \ourmethod 

We train XGBoost classification models separately on \(D_o\) and \(D_w\). Then we evaluate their classification performance. The $F_1$-scores for categories 2, 4, and 6 are presented in Table ~\ref{tab:rob1_class}. It can be observed that the ML utility of $D_w$, compared to $D_o$, does not significantly decrease, which implies that \ourmethod hardly impairs the ML utility of the watermarked dataset. This indicates that \ourmethod is non-intrusive. 

\begin{table}[h!]
  \centering
  \caption{$\bm{F_1}$-scores.}
  \label{tab:rob1_class}
  \begin{tabular}{@{}>{\centering\arraybackslash}p{1.5cm}>{\centering\arraybackslash}p{1.8cm}>{\centering\arraybackslash}p{1.8cm}>{\centering\arraybackslash}p{1.8cm}@{}}
    \toprule
    Dataset & Category 2 & Category 4 & Category 6  \\
    \midrule
    $D_o$ & 0.888 & 0.940 & 0.848 \\
    $D_w$ & 0.887 & 0.937 & 0.845 \\
    \bottomrule
  \end{tabular}
\end{table}


\partitle{More Datasets} 
We continue to verify the non-intrusiveness of \ourmethod on the HOG feature dataset and Boston Housing Prices dataset. Random Forest and linear regression models are employed to assess their ML utility. The experimental results are as follows. 
 

 \begin{table}[h] 
    \centering
    \caption{ML utility on more datasets.}\label{tab:rob2_o_w}
    \begin{minipage}{0.45\linewidth}
      \centering
      \text{(a) Accuracy on \textit{HOG}.}\label{subtab:omc}\\
        \begin{tabular}{ccc}
        \toprule
        Dataset & $D_o$ & $D_w$ \\
        \midrule
        Accuracy & 0.942 & 0.940 \\
        \bottomrule
      \end{tabular}
    \end{minipage}
    \begin{minipage}{0.45\linewidth}
      \centering
          \text{(b) MSEs on \textit{Boston Housing}.}\label{subtab:nmc}\\
          \begin{tabular}{ccc}
            \toprule
            Dataset & $D_o$ & $D_w$ \\
            \midrule
            MSE & 24.8 & 25.6 \\
            \bottomrule
          \end{tabular}
    \end{minipage}
 \end{table}

The experimental results show that the ML-related metrics, including MSE and accuracy, are slightly reduced after watermarking, demonstrating the non-intrusiveness of \ourmethod. 



\subsection{Robustness} \label{subsec:rq3}
In this section, we investigate the z-scores and ML-related metrics on the attacked datasets to answer \textbf{RQ3}. Following the proposed schemes~\cite{agrawal2002watermarking,hu2018new,li2020robust}, we mainly verify the robustness of \ourmethod on the real-world dataset Forest Cover Type. 

\label{subsec:4.3.1}
\partitle{Alteration Attack} An alteration attack on a dataset refers to a type of attack where an attacker intentionally modifies the attributes of the dataset in a random or seemingly random manner. The randomness of the changes may make it difficult to detect watermarks. In this experiment, we assume that the attacker is aware of the attribute to which the watermark has been embedded (in actual scenarios, the attribute used for watermarking is generally kept secret). We simulate the attacker randomly perturbing \(20\%\), \(40\%\), \(60\%\), \(80\%\) and \(100\%\) of the attribute {Cover Type} and examine the z-scores under various alteration attacks along with the ML utility impact on the XGBoost model. Table~\ref{tab:rob1_alt_z} and Table~\ref{tab:rob1_alt_xgb} respectively show the results of the z-scores and the classification performance for the dataset under alteration attacks. 

\begin{table}[h!]
  \centering
  \caption{Z-scores.}
  \label{tab:rob1_alt_z}
  \begin{tabular}{@{}cccccc@{}}
    \toprule
    Alteration Attack ($$\%$$) & 20 & 40 & 60 & 80 & 100 \\
    \midrule
    Z-score & 13.8 & 8.67 & 3.68 & -1.11 & -6.05 \\
    \bottomrule
  \end{tabular}
\end{table}

\begin{table}[h!]
  \centering
  \caption{$\bm{F_1}$-scores.}
  \label{tab:rob1_alt_xgb}
  \begin{tabular}{@{}ccccccc@{}}
    \toprule
    Alteration Attack ($\%$) & Category 2 & Category 4 & Category 6  \\
    \midrule
    20 & 0.879 & 0.853 & 0.813 \\
    40 & 0.867 & 0.791 & 0.761 \\
    60 & 0.847 & 0.661 & 0.645 \\
    80 & 0.679 & 0.0631 & 0.115 \\
    100 & 0.00203 & 0.00180 & 0.000773 \\
    \bottomrule
  \end{tabular}
\end{table}

We can observe that, given \( \alpha = 1.96 \), when the proportion of the alteration attack reaches 80\%, the watermark detection algorithm can no longer detect the watermark in the attacked dataset. However, at this point, the $F_1$-scores for these three categories decrease significantly, indicating that the ML utility of the attacked dataset for the model has been greatly reduced. Although the attacker successfully erases the watermark, the trained model becomes inoperable, which implies that \ourmethod is robust under alteration attacks.

\partitle{Insertion Attack} An insertion attack on a dataset refers to a type of attack where tuples are created randomly and inserted into the dataset to destroy the watermark. In this experiment, we simulate the attacker randomly inserting \(20\%\), \(40\%\), \(60\%\), \(80\%\), \(100\%\) tuples into the watermarked dataset and examine the z-scores under various insertion attacks. Due to the insertion attack disrupting the order of tuples in the dataset, we employ the first two attributes and the first three attributes of the Forest Cover Type as primary keys respectively to match tuples. Table~\ref{tab:ins_z_2} and Table~\ref{tab:ins_z_3} display the z-score detection results for these two matching methods, respectively.

\begin{table}[h!]
  \centering
  \caption{Z-scores (two-attribute matching).}
  \label{tab:ins_z_2}
  \begin{tabular}{@{}cccccc@{}}
    \toprule
   Insertion Attack ($$\%$$) & 20 & 40 & 60 & 80 & 100 \\
    \midrule
    Z-score & 8.23 & 7.43 & 6.87 & 6.24 & 5.73 \\
    \bottomrule
  \end{tabular}
\end{table}

\begin{table}[h!]
  \centering
  \caption{Z-scores (three-attribute matching).}
  \label{tab:ins_z_3}
  \begin{tabular}{@{}cccccc@{}}
    \toprule
   Insertion Attack ($$\%$$) & 20 & 40 & 60 & 80 & 100 \\
    \midrule
    Z-score & 18.3 & 18.1 & 18.1 & 17.9 & 17.8\\
    \bottomrule
  \end{tabular}
\end{table}

The results from the two tables indicate that, under the assumption of a threshold \( \alpha = 1.96 \), watermarks can still be effectively detected even when up to 100\% of tuples are inserted relative to the size of the original dataset, which signifies that \ourmethod is robust under insertion attacks.

\partitle{Deletion Attack} A deletion attack on a dataset refers to a type of attack where tuples are deleted randomly to destroy the watermark. In this experiment, we simulate the attacker randomly deleting \(20\%\), \(40\%\), \(60\%\), \(80\%\), \(100\%\) tuples from the watermarked dataset and examine the z-scores under various deletion attacks. Due to the same reason with the insertion attack, we also employ. the first two attributes and the first three attributes of the Forest Cover Type as primary keys respectively to match tuples. Tables~\ref{tab:del_z_2} and ~\ref{tab:del_z_3} display the z-score detection results for these two matching methods, respectively. 

\begin{table}[h!]
  \centering
  \caption{Z-scores (two-attribute matching).}
  \label{tab:del_z_2}
  \begin{tabular}{@{}cccccc@{}}
    \toprule
   Deletion Attack ($$\%$$) & 20 & 40 & 60 & 80 & 100 \\
    \midrule
    Z-score & 8.42 & 7.45 & 6.13 & 4.59 & $\backslash$ \\
    \bottomrule
  \end{tabular}
\end{table}

\begin{table}[h!]
  \centering
  \caption{Z-scores (three-attribute matching).}
  \label{tab:del_z_3}
  \begin{tabular}{@{}cccccc@{}}
    \toprule
    Deletion Attack ($$\%$$) & 20 & 40 & 60 & 80 & 100\\
    \midrule
    Z-score & 16.7 & 14.4 & 11.6 & 8.21 & $\backslash$ \\
    \bottomrule
  \end{tabular}
  \vspace{0.32cm}
\end{table}

The results from the two aforementioned tables indicate that, under the assumption of a threshold \( \alpha = 1.96 \), the watermark can still be effectively detected even in the event of a substantial deletion attack, which suggests that \ourmethod is robust under deletion attacks.

\partitle{Shuffle Attack} A shuffle attack on tabular datasets is a form of attack where the attacker rearranges or randomizes the order of rows in a table, thus disrupting the original sequence of the data. This attack specifically targets watermarking schemes that rely on row order, with the intent to the detectability of the watermark without compromising the utility of datasets. Shuffle attacks share many similarities with the insertion and deletion attacks mentioned previously. Using the matching algorithm within the watermark detection method, watermarks on datasets subjected to such attacks can be effectively identified.

\partitle{Attributes Subset Attack} In some scenarios, attackers may choose to take a subset of attributes of a dataset to disrupt the watermark. This can be mitigated in the watermarking scheme, as we choose to embed the watermark in the last attribute (target), making it infeasible for attackers to erase the watermark by simply removing this attribute. Furthermore, we have the option to add watermarks across multiple attributes. Such attacks might still pose challenges to the matching algorithm. Nevertheless, the experiments mentioned above demonstrate that even when only two attributes are used as a primary key for matching, the detection efficacy remains high.

\partitle{Invertibility Attack} An invertibility attack on tabular datasets targets the watermarking schemes to revert the dataset to its unmarked state while preserving data utility. For instance, if an attribute in a dataset is used for embedding a watermark, an attacker might employ statistical analysis to identify and modify or remove the anomalous values that do not conform to the expected statistical model, which is likely to represent the watermark information. If the attacker successfully identifies and reverses these modifications, the watermark may be invalid without significantly compromising the utility of data. \ourmethod requires only a small number of tuples relative to the total volume of the dataset for generating the watermark. For instance, in the aforementioned experiment, the Forest Cover Type dataset contains over half a million records, yet we only needed 400 tuples to embed the watermark. Such a subtle modification renders the task of identifying and eliminating the watermark exceedingly difficult for attackers.



\partitle{More Datasets} We continue to validate the robustness of \ourmethod on the HOG feature dataset and Boston Housing Prices Dataset. From Section~\ref{subsec:4.3.1}, due to our method being difficult to remove by other types of attacks, we only simulate alteration attacks.

\begin{table}[h!]
  \centering
  \caption{Z-scores and accuracy on HOG.}
  \label{tab:rob2_alter}
  \begin{tabular}{@{}cccccc@{}}
    \toprule
    Alteration Attack ($$\%$$) & 20 & 40 & 60 & 80 & 100 \\
    \midrule
    Z-score & 9.71 & 7.13 & 4.78 & 2.55 & -0.0816  \\
    Accuracy & 0.933 & 0.910 & 0.825 & 0.530 & 0.100  \\
    \bottomrule
  \end{tabular}
\end{table}

\begin{table}[h!]
  \centering
  \caption{Z-scores and MSEs on Boston Housing.}
  \label{tab:rob3_alter}
  \begin{tabular}{@{}cccccc@{}}
    \toprule
    Alteration Attack ($$\%$$) & 20 & 40 & 60 & 80 & 100 \\
    \midrule
    Z-score & 5.30 & 3.76 & 2.17 & 0.571 & -1.10  \\
    MSE & 26.1 & 27.6 & 34.0 & 35.4 & 42.3  \\
    \bottomrule
  \end{tabular}
\end{table}

 The results show that attackers would incur substantial costs to make the watermark undetectable, a cost so high that it renders the trained model nearly unusable. This means that \ourmethod can efficiently safeguard the ownership of these datasets.


\subsection{Comparison} \label{subsec:rq4}
In this section, we compare related watermarking schemes with \ourmethod to answer \textbf{RQ4}. The real-world dataset Forest Cover Type is utilized to verify \ourmethod outperforms related work.

Due to the inconsistency of the detection metrics of \ourmethod with related relational database watermarking works, direct comparison with related work is not feasible. Most other works employ the bit error rate (BER) to measure the integrity of the extracted watermark information. In contrast, we use the z-score of the one proportion z-test as our detection metric. As we do not embed watermark plaintext, we use the proportion of key cells flipped from green domains to red domains as a metric similar to the BER. These metrics are uniformly represented as {Mismatch Percentage} (MP). Although this proportion does not directly indicate the strength of the watermark, it is the closest method available for comparison with metrics used in related work.

We consider two recent relational database watermarking schemes for comparison with \ourmethod: 1) HistMark: \citet{hu2018new} utilize histogram shifting to embed watermark information.  2) SemMark: \citet{li2022secure}  alter the random least significant position of the decimal numerical attributes to embed watermark information. Both methods are considered state-of-the-art reversible watermarking schemes and are capable of effectively controlling distortion. We will implement \ourmethod, HistMark, and SemMark on the Forest Cover Type dataset and compare them based on the three goals of a watermarking scheme. We use a subset of the Forest Cover Type dataset to avoid the adverse impact of having a much larger number of instances in certain categories on watermarking schemes. For the experiments, we randomly select 2000 tuples from each category. The number of key cells for \ourmethod and the length of the watermark bit sequence is set to 400, ensuring that the quantity of 0s and 1s in the watermark bit sequence is equal. The watermark is embedded in the attribute Cover\_Type. The classification model XGBoost is utilized to test the ML utility. Figure~\ref{fig:comp_f1} and Figure~\ref{fig:comp_mp} exhibit the results of the comparative experiments.

\begin{figure}
    \begin{subfigure}{0.25\textwidth}
        \includegraphics[width=\linewidth]{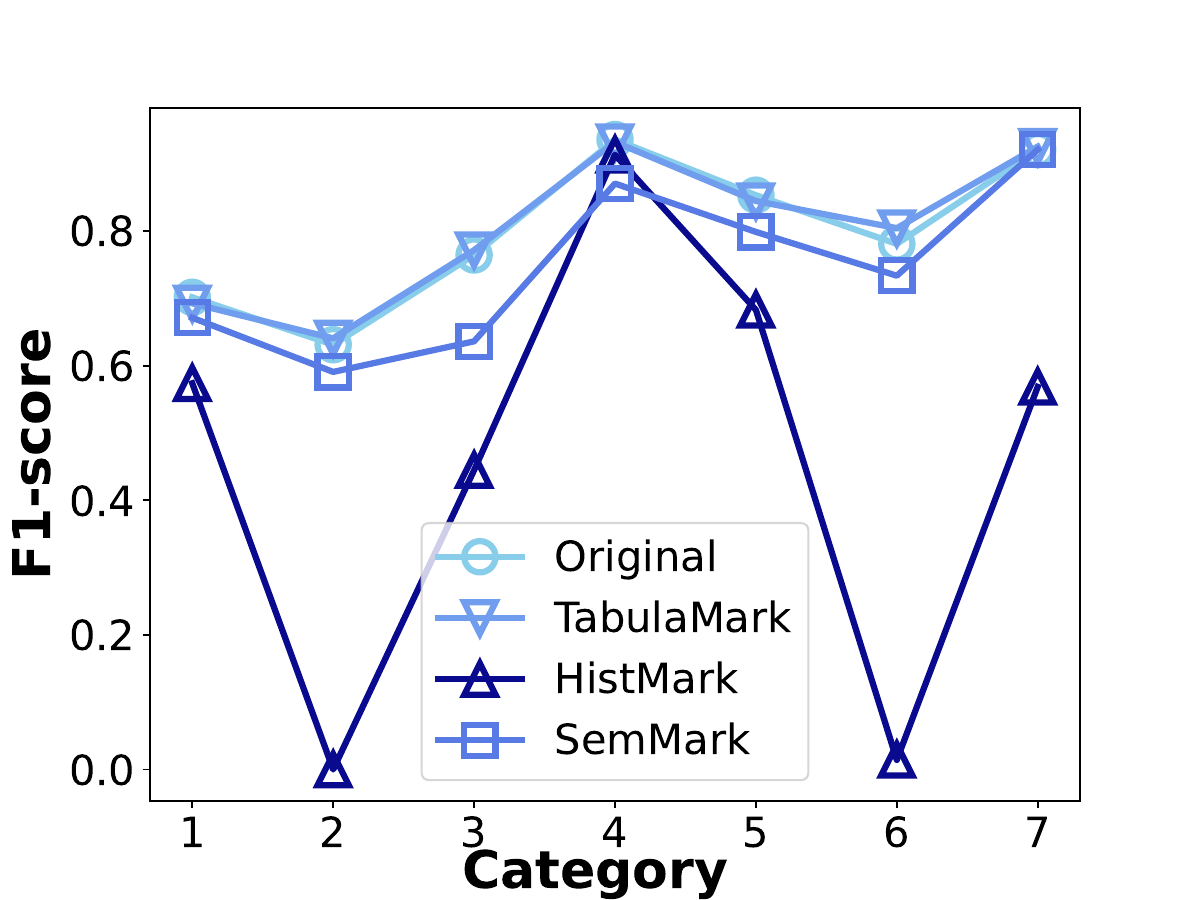}
        \caption{Accuracy}
        \label{fig:comp_f1}
    \end{subfigure}%
    \begin{subfigure}{0.25\textwidth}
        \includegraphics[width=\linewidth]{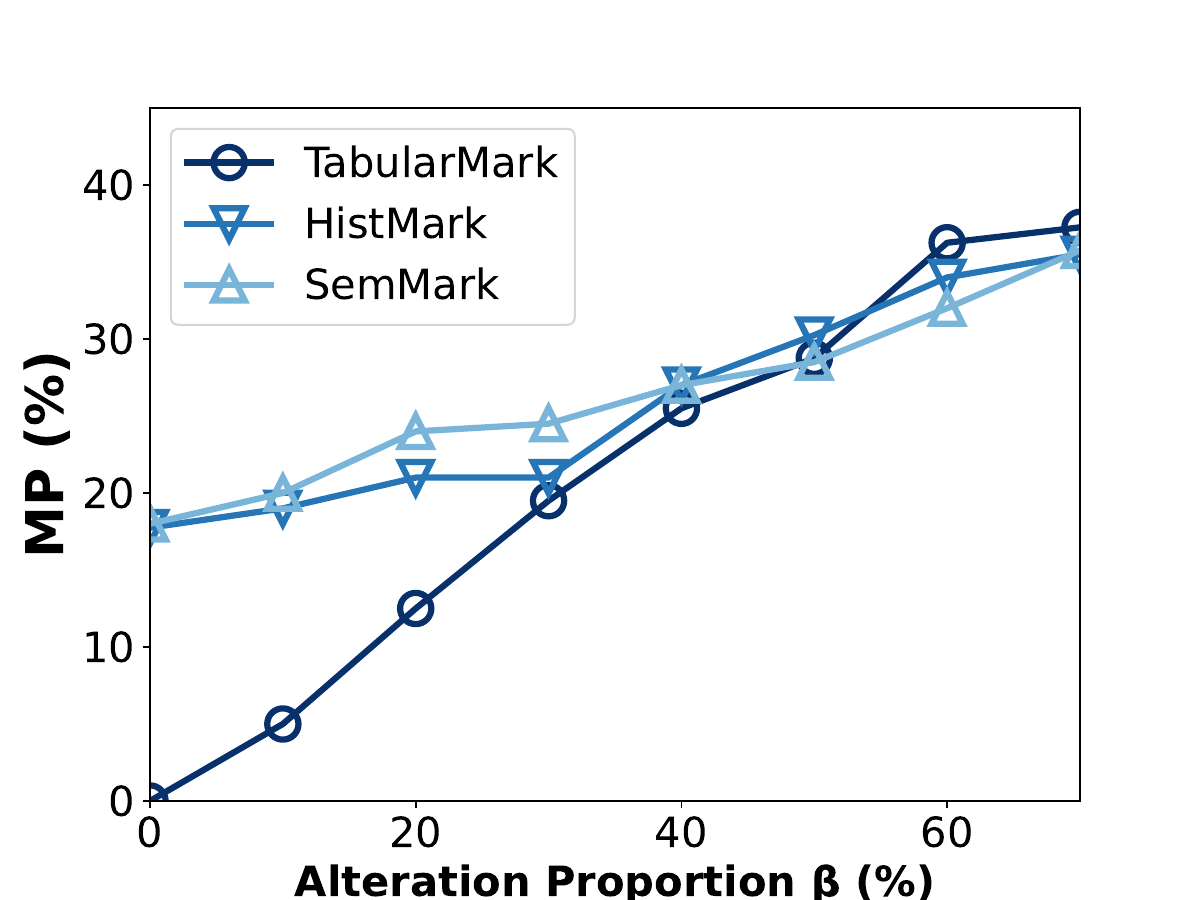}
        \caption{Mismatch percentage}
        \label{fig:comp_mp}
    \end{subfigure}
    \caption{Comparisons of \ourmethod (ours), HistMark, and SemMark on (a) non-intrusiveness and (b) robustness.}
\end{figure}

 Compared to the original dataset, the watermarked dataset of \ourmethod exhibits a negligible decline in the $F_1$-scores across various categories, with the lines almost overlapping. In contrast, the line for SemMark presents a noticeable gap relative to \ourmethod. This discrepancy is due to the requirement of SemMark to ensure successful watermark embedding and the use of a majority voting mechanism for robustness, which introduces a higher degree of data distortion. The outcomes for HistMark are even less satisfactory, which is due to histogram shifting that alters a large amount of data in every group. Even changes as minimal as $\pm 1$ can introduce excessive data distortion for categorical attributes.

 When the attack proportion is set to 0, the MP for all three watermarking schemes tends to be low, indicating that detectability is generally assured. However, the mechanisms employed by SemMark and HistMark may result in some watermark bits not being successfully embedded into the dataset. For HistMark, if the peak bin of the histogram of each group matches the maximum or minimum value of the corresponding attribute, it may not be possible to shift it, leading to the failure of watermark embedding for that group. In the case of SemMark, the location where a watermark bit is embedded within a tuple is randomly selected, which may cause some positions to be unselected. \ourmethod, on the other hand, does not rely on the embedding of specific watermark bits and is not susceptible to this issue.

 As shown in Figure~\ref{fig:comp_mp}, when the alteration attack proportion increases, aside from an initially slightly lower MP for \ourmethod, the plots of all three watermarking schemes converge beyond a 30\% proportion, suggesting that there is no significant disparity in the MP metric among the schemes. Although \ourmethod employs a threshold-based z-score for watermark detection rather than MP, this implies that the robustness of \ourmethod is comparable to that of state-of-the-art methods to a certain extent.

\subsection{Design Trade-Offs} \label{subsec:rq5}
In this section, we conduct a qualitative analysis of the crucial parameters of \ourmethod and analyze their trade-offs to answer \textbf{RQ5}. Synthetic datasets are utilized to explore trade-offs to exclude external factors' interference.

To generate labels for classification experiments, data across two dimensions is combined with weights in the range of \([-1, 1]\), and the weighted sum is passed through a logistic function to compute the probability of an event occurring. This probability is then used to conduct a binomial distribution experiment to generate a target array.  The numerical attribute Dimension 1 is utilized to investigate the trade-offs of the parameters. A logistic regression model is trained to evaluate data utility. We assume that the attacker employs an alteration attack, aiming to modify the dataset subtly in an attempt to remove the watermark without making overt changes that would degrade the ML utility of the dataset. We employ a logistic regression model as the classifier for experiments.

\subsubsection{The Perturbation Range $[-p,p]$} \label{subsubsec:range_p}
The parameter \( p \) is related not only to the strength of the watermark but also to the ML utility of the dataset. Here, we treat \(p\) as a variable, with values set at  \( 0.5\sigma \), \( 1\sigma \), \( 1.5\sigma \), \( 2\sigma \), and \( 2.5\sigma \). Other parameters are set as \( k = 500 \) and \( n_w = 300 \). We fix the alteration range at \(2\sigma\), meaning that they perturb the dataset within a range of \([-2\sigma, 2\sigma]\). The attack proportions are set at \( 20\%, 40\%, 60\%, 80\%, \) and \( 100\% \). 

In the following, we denote these proportions using the symbol \( \beta \). Table~\ref{tab:X_o_w} illustrates the impact of watermarking on the ML utility of the dataset, measured by the classification accuracy of the trained logistic regression model. Figure~\ref{fig:zscore_variation_p} and Figure~\ref{fig:accuracy_variation_p} respectively illustrate the variations in z-scores and classification accuracy of the dataset under attacks as the parameter \( p \) changes. Table~\ref{tab:p_z_a} presents the specific values of z-scores and classification accuracy for the attacked datasets.

\begin{table}[h!]
  \centering
  \caption{Accuracy.}
  \label{tab:X_o_w}
  \begin{tabular}{@{}ccccccc@{}}
    \toprule
    Dataset & \( 0.5\sigma \)  & \( 1.0 \sigma \) & \( 1.5\sigma \) & \( 2.0\sigma \) & \( 2.5\sigma \) \\
    \midrule
    \(D_o\) & 0.977 & 0.977 & 0.977 & 0.977 & 0.977\\
    \(D_w\) & 0.976 & 0.973 & 0.965 & 0.956 & 0.952\\
    \bottomrule
  \end{tabular}
\end{table}

\begin{figure}[h]
    \centering
    \begin{subfigure}[b]{0.23\textwidth}
        \includegraphics[width=\linewidth]{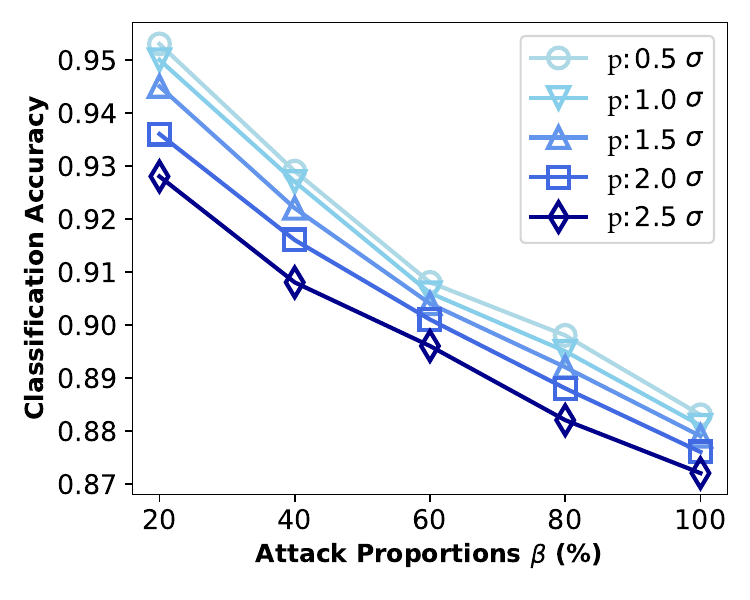}
        \caption{Accuracy}
        \label{fig:zscore_variation_p}
    \end{subfigure}%
    \hspace{.05in}
    \begin{subfigure}[b]{0.23\textwidth}
        \includegraphics[width=\linewidth]{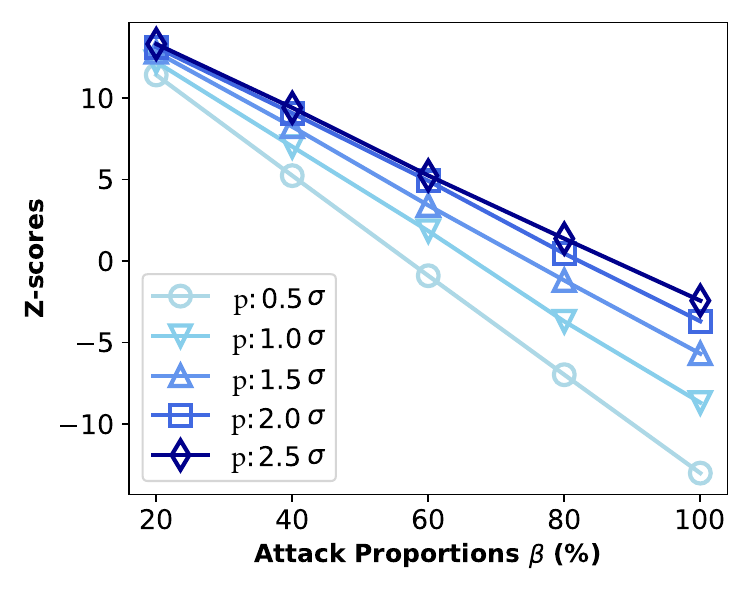}
        \caption{Z-scores}
        \label{fig:accuracy_variation_p}
    \end{subfigure}
    \caption{Accuracy and z-scores on \(\beta\)s.}
\end{figure}

\begin{table}[t]
 \caption{Z-scores and accuracy on \(p\)s.}
 \noindent\adjustbox{max width=\linewidth}{%
 \centering
 \begin{tabular}{lcccccccccc}\toprule
  \setlength{\tabcolsep}{10pt} 
\multirow{2}{*}{$p$} & \multicolumn{5}{c}{Z-score} & \multicolumn{5}{c}{Accuracy} \\
 \cmidrule(lr){2-6}\cmidrule(lr){7-11}
  & 20 & 40 & 60 & 80 & 100 & 20 & 40 & 60 & 80 & 100
\\
\midrule
    $0.5\sigma$ & 11.4 & 5.23 & -0.903 & -6.97 & -13.0 & 0.953 & 0.929 & 0.908 & 0.898 & 0.883
\\
    $1.0\sigma$ & 12.2 & 7.01 & 1.84 & -3.70 & -8.68 & 0.950 & 0.927 & 0.906 & 0.895 & 0.881
\\
    $1.5\sigma$ & 12.8 & 8.23 & 3.41 & -1.19 & -5.68 & 0.945 & 0.922 & 0.904 & 0.892 & 0.879
\\
    $2.0\sigma$ & 13.1 & 9.06 & 4.91 & 0.457 & -3.69 & 0.936 & 0.916 & 0.901 & 0.888 & 0.876
\\
    $2.5\sigma$ & 13.3 & 9.39 & 5.24 & 1.37 & -2.43 & 0.928 & 0.908 & 0.896 & 0.882 & 0.872
\\
  \bottomrule
  \end{tabular}
}
 \label{tab:p_z_a}
\end{table}

With the increase in \( p \), the ML utility of the watermarked dataset degrades. In alteration attacks, it can be observed that as \( \beta \) increases, the z-score decreases rapidly. Meanwhile, the decrease in the utility of attacked datasets is more pronounced. With the increase in \( p \), the ML utility of the attacked dataset degrades. This is because as \( p \) increases, the perturbed data are more likely to fall within the green domains, forcing attackers to increase the rate of alteration, resulting in making the watermark unusable. The larger \( \beta \), the poorer the performance of the trained model. It is important to note that \( p \) cannot be increased indefinitely, as key cells may be identified as outliers.

\subsubsection{The Number of Key Cells $n_w$}
The number of key cells, denoted by \( n_w \), is also a crucial parameter. Here, we treat \( n_w \) as a variable, with values set at 100, 200, 300, 400, and 500. Other parameters are set as \( p = 2.0\sigma \) and \( k = 500 \). We also fix the alteration range of perturbation at \(2\sigma\). The attack proportions are set at \( 20\%, 40\%, 60\%, 80\%, \) and \( 100\% \). Table~\ref{tab:n_o_w} illustrates the impact of watermarking on the ML utility of the dataset, measured by the classification accuracy of the trained logistic regression model. Figure~\ref{fig:zscore_variation_n} and Figure~\ref{fig:accuracy_variation_n} respectively illustrate the variations in z-scores and classification accuracy of the dataset under attack as the parameter \( n_w \) changes. Table~\ref{tab:n_z_a} presents the specific values of z-scores and classification accuracy for the attacked datasets.

\begin{table}[h!]
  \centering
  \caption{Accuracy.}
  \label{tab:n_o_w}
  \begin{tabular}{@{}ccccccc@{}}
    \toprule
    Dataset & 100  & 200 & 300 & 400 & 500 \\
    \midrule
    \(D_o\) & 0.977 & 0.977 & 0.977 & 0.977 & 0.977\\
    \(D_w\) & 0.972 & 0.963 & 0.957 & 0.953 & 0.949\\
    \bottomrule
  \end{tabular}
\end{table}

\begin{figure}
    \begin{subfigure}{0.23\textwidth}
        \includegraphics[width=\linewidth]{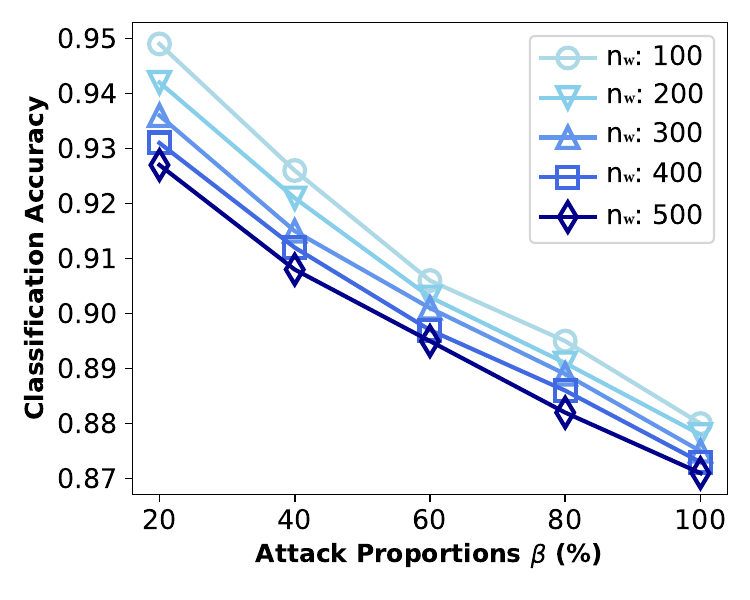}
        \caption{Accuracy}
        \label{fig:zscore_variation_n}
    \end{subfigure}%
    \hspace{.05in}
    \begin{subfigure}{0.23\textwidth}
        \includegraphics[width=\linewidth]{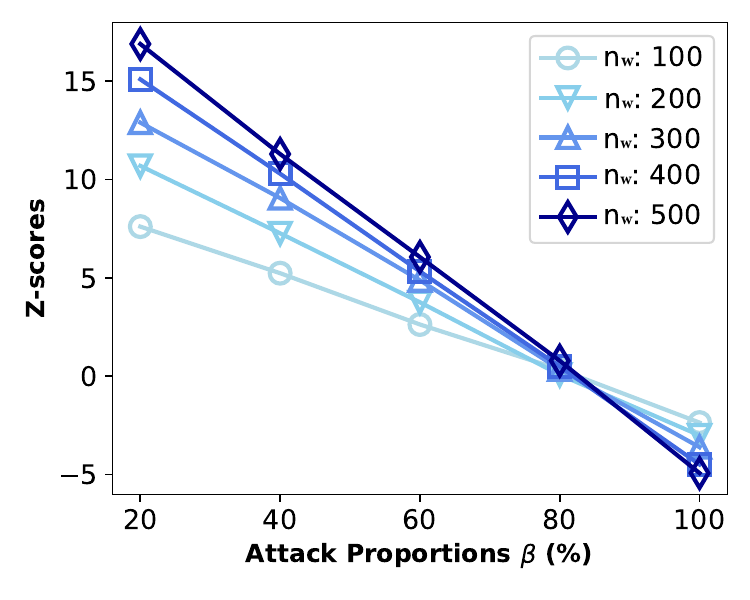}
        \caption{Z-scores}
        \label{fig:accuracy_variation_n}
    \end{subfigure}
    \caption{Accuracy and z-scores on \(\beta\)s.}
\end{figure}

\begin{table}[t]
 \caption{Z-scores and accuracy on \(n_w\)s.}
 \centering
\noindent\adjustbox{max width=\linewidth}{%
 \begin{tabular}{lcccccccccc}\toprule
    \setlength{\tabcolsep}{10pt} 
   \multirow{2}{*}{$n$} & \multicolumn{5}{c}{Z-score} & \multicolumn{5}{c}{Accuracy} \\
    \cmidrule(lr){2-6}\cmidrule(lr){7-11}
    & 20 & 40 & 60 & 80 & 100 & 20 & 40 & 60 & 80 & 100
\\
\midrule
    100 & 7.61 & 5.24 & 2.63 & 0.344 & -2.36 & 0.949 & 0.926 & 0.906 & 0.895 & 0.880 
\\
    200 & 10.7 & 7.27 & 3.76 & 0.0735 & -3.01 & 0.942 & 0.921 & 0.903 & 0.891 & 0.878 
\\
    300 & 12.9 & 9.06 & 4.88 & 0.351 & -3.61 & 0.936 & 0.915 & 0.901 & 0.889 & 0.875 
\\
    400 & 15.1 & 10.3 & 5.33 & 0.496 & -4.49 & 0.931 & 0.912 & 0.897 & 0.886 & 0.873 
\\
    500 & 16.9 & 11.3 & 6.06 & 0.785 & -4.93 & 0.927 & 0.908 & 0.895 & 0.882 & 0.871 
\\
  \bottomrule
 \end{tabular}
 }
 \label{tab:n_z_a}
\end{table}

It can be observed that as \( n_w \) increases, the ML utility of the watermarked dataset declines due to the increase of perturbed data, and there is a slight decrease in the classification performance of the trained model. Simultaneously, the cost of attack rises, as a larger \( n_w \) results in a higher z-score under the same attack intensity, making the watermark less susceptible to being erased.

\subsubsection{The Proportion of Green Domains \( \gamma \)}
This represents a unique trade-off, as in the watermark embedding and detection algorithms, the ratio of the length of green domains to the total length of the perturbation range, represented as \( \gamma \), is set at 0.50. This decision is based on the assumption that the deviation of cells in suspicious datasets from those in the original dataset falls within a certain range randomly. However, the impact of variations in \( \gamma \) on \ourmethod is worth exploring. Here, we treat \(gamma\) as a variable, with values set at \(0.25, 0.33, 0.5, 0.67, 0.75\), \(n_w = 300\). We also fix \(p = 1.5 \sigma \) and \(k = 500\). 
We fix the alteration range at \(1.5 \sigma\). Table~\ref{tab:gamma_o_w} illustrates the impact of watermarking on the ML utility of the dataset, measured by the classification accuracy of the trained logistic regression model. Figures~\ref{fig:zscore_variation_gamma} and~\ref{fig:accuracy_variation_gamma} respectively illustrate the variations in z-scores and classification accuracy of the dataset under attack as the parameter \( \gamma \) changes. It can be observed that, as \( \gamma \) increases, the ML utility of the watermarked dataset remains nearly unchanged, indicating that variations in \( \gamma \) hardly alter the degree of data distortion. When attackers increase the attack proportion, the classification accuracy for different \( \gamma \) values tends to converge, while smaller \( \gamma \) values result in larger z-scores. This is dictated by the calculation formula of the z-score, indicating that smaller \( \gamma \) values can enhance the robustness of \ourmethod.

\begin{table}[h!]
  \centering
  \caption{Accuracy.}
  \label{tab:gamma_o_w}
  \begin{tabular}{@{}ccccccc@{}}
    \toprule
    Dataset & 0.25  & 0.33 & 0.50 & 0.67 & 0.75 \\
    \midrule
    \(D_o\) & 0.958 & 0.958 & 0.958 & 0.958 & 0.958\\
    \(D_w\) & 0.958 & 0.958 & 0.956 & 0.957 & 0.957\\
    \bottomrule
  \end{tabular}
\end{table}

\begin{figure}
    \begin{subfigure}{0.23\textwidth}
        \includegraphics[width=\linewidth]{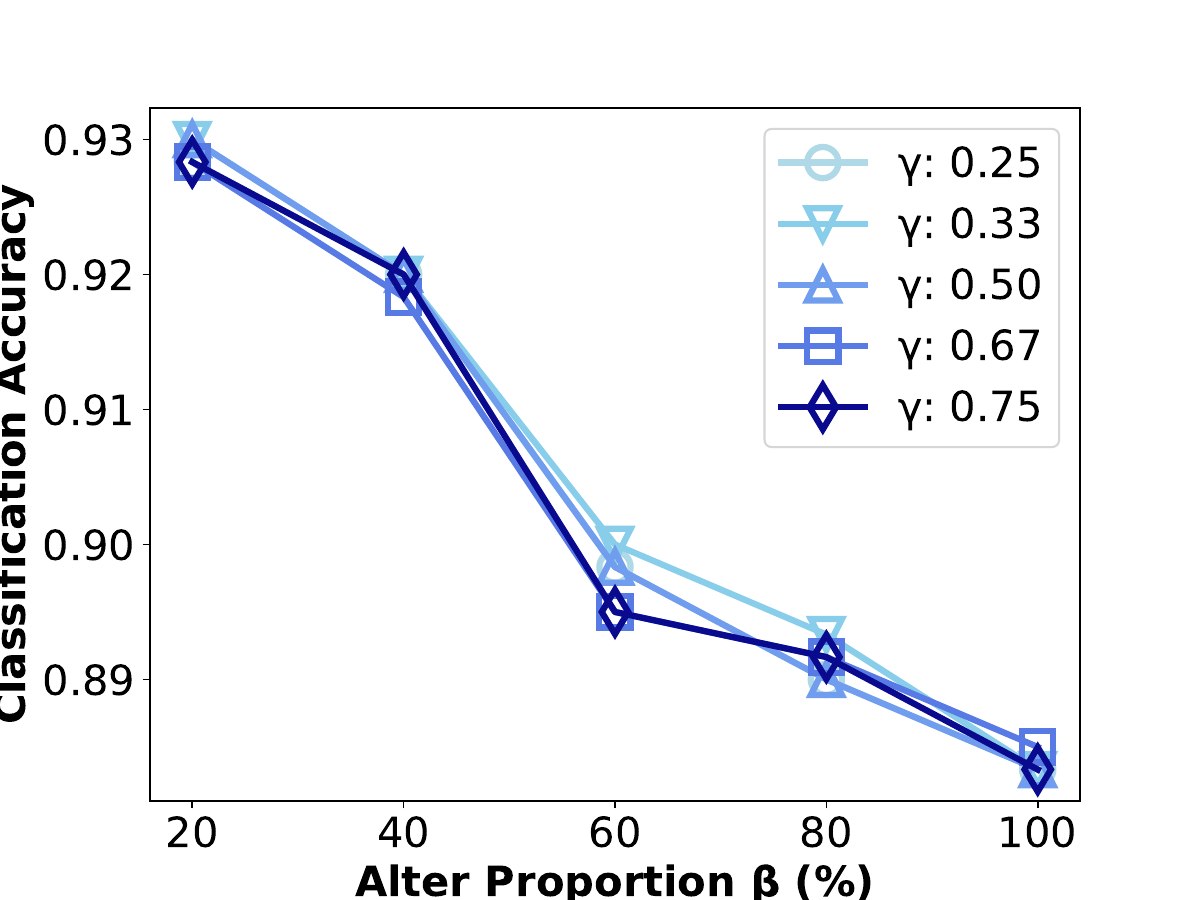}
        \caption{Accuracy.}
        \label{fig:zscore_variation_gamma}
    \end{subfigure}%
    \hspace{.05in}
    \begin{subfigure}{0.23\textwidth}
        \includegraphics[width=\linewidth]{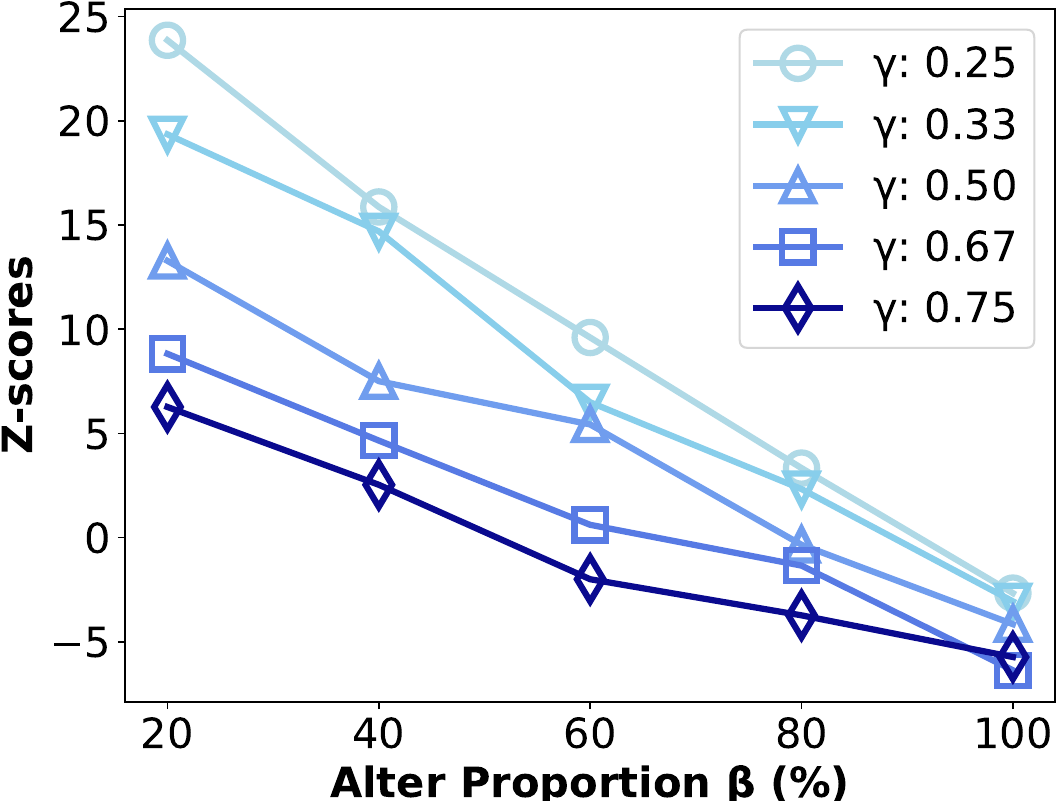}
        \caption{Z-scores.}
        \label{fig:accuracy_variation_gamma}
    \end{subfigure}
    \caption{Accuracy and z-scores on \(\beta\)s.}
\end{figure}



To further investigate the trade-off afforded by \( \gamma \), we tested the z-scores of different watermarking schemes on the original dataset. The average results are presented in Table~\ref{tab:gamma_z_original}. It can be observed that when \( \gamma \) is either very small or very large, the z-scores tend to be larger. This suggests that when the threshold is low, there is a potential risk of false positives, where an unwatermarked original dataset might be misidentified as watermarked. This occurs because if the green cells of the original dataset deviate slightly from the expected values when \( \gamma \) is small, it can result in a disproportionately large z-score. Conversely, when \( \gamma \) is large, the key cells of the original dataset naturally fall within the green domains with a higher probability, leading to larger z-scores as well. Therefore, selecting \( \gamma = 0.50\) may represent a more balanced choice.

\subsubsection{Conclusion}
Increasing the size of $p$ and $n_w$ can enhance the robustness of the watermark, but it may compromise non-intrusiveness. We suggest that data owners adjust these two parameters based on their actual needs. Additionally, keeping \( \gamma \) at \( \frac{1}{2} \) helps to reduce the potential risk of false positives so we suggest data owners use it.

\begin{table}[h!]
  \centering
  \caption{Z-scores.}
  \begin{tabular}{@{}ccccccc@{}}
    \toprule
    Dataset & 0.25  & 0.33 & 0.50 & 0.67 & 0.75 \\
    \midrule
    \(D_o\) & 1.47 & 0.489 & 0.346 & 0.122 & 0.933\\
    \bottomrule
  \end{tabular}
  \label{tab:gamma_z_original}
\end{table}

\blueno{}


\subsection{Discussion} \label{subsec:rq6}
In this section, we discuss an optimization strategy with a new perturbation noise selection method to enhance the non-intrusiveness of \ourmethod. 

In Line 6 of Algorithm~\ref{alg:watermark_generation_numerical}, we randomly select a number from the green domains for perturbation. However, this uniform random selection strategy may pick up a fair amount of substantial noise to perturb the original dataset, which could potentially have a noticeable impact on the ML utility of the data. To address this issue, we improve the original strategy by adopting the normal distribution probability. This means that numbers closer to the original value have a higher probability of being selected, which is intuitively reasonable as it implies that larger deviations are less likely to occur. We characterize this transition with the corresponding PDF. In the watermark embedding phase, we continuously sample from the probability density function on \([-p, p]\) until the sample falls within the green domains.

The PDF for uniform random selection is given by:


\begin{equation*}
f(x) = 
\begin{cases} 
  \frac{1}{2p}, & \text{if } -p \leq x \leq p,  \\
  0, & \text{otherwise}.
\end{cases}
\label{eq:uniform_pdf}
\end{equation*}

The PDF for selection based on a normal distribution is given as:

\begin{equation*}
f(x) = 
\begin{cases} 
  \frac{1}{\sigma\sqrt{2\pi}} \exp\left(-\frac{(x-\mu)^2}{2\sigma^2}\right) \bigg/ \left(\Phi(p) - \Phi(-p)\right), & \text{if } -p \leq x \leq p, \\
  0, & \text{otherwise}.
\end{cases}
\label{eq:normalized_pdf}
\end{equation*}

We continue experiments on the Boston Housing Prices dataset. Apart from the strategy for noise selection, all other experimental settings remain the same as those described in Subsection~\ref{subsec:rq3}. We investigate whether the ML utility of the new watermarked dataset, with the improved strategy for noise selection, shows improvement compared to the original watermark dataset. Additionally, we examine the robustness of the new watermark dataset against alteration attacks, with the attack strategy being the same as before. Table~\ref{tab:non-random_per} displays the results for the original dataset, the watermarked dataset with uniform random perturbation, and the watermarked dataset with random perturbation based on a normal distribution probability.  Table~\ref{tab:non-random_per_attacked} displays the results for the new watermark dataset under alteration attacks.

\begin{table}[h!]
  \centering
  \caption{MSEs on \(D_o\) and \(D_w\)s .}
  \label{tab:non-random_per}
  \begin{tabular}{@{}cccc@{}}
    \toprule
    Dataset & $D_o$ & Completely Random & Normal Distribution \\
    \midrule
    MSE & 24.8 & 25.6  & 25.0\\
    \bottomrule
  \end{tabular}
\end{table}

\begin{table}[h!]
  \centering
  \caption{Z-scores and MSEs.}
  \label{tab:non-random_per_attacked}
  \begin{tabular}{@{}ccccccc@{}}
    \toprule
    Alteration Attack ($$\%$$) & 20 & 40 & 60 & 80 & 100 \\
    \midrule
    Z-score & 5.46 & 3.88 & 2.18 & 0.794 & -0.690  \\
    MSE & 25.9 & 27.1 & 34.9 & 33.8 & 41.8  \\
    \bottomrule
  \end{tabular}
  \vspace{0.5cm}
\end{table}

The experimental results indicate that perturbing data based on sampling from a normal distribution probability effectively reduces the impact of watermarking on the ML utility of the dataset while still maintaining robustness.
\vspace{+1em}

\section{Conclusion}\label{sec:conclusion}
In this paper, we propose a simple yet effective watermarking scheme for tabular datasets in ML, \ourmethod, showing superior detectability, non-intrusiveness, and robustness. \ourmethod is inspired by the randomness in the deviation between a suspicious dataset and the corresponding original dataset. When embedding, \ourmethod artificially divides the data with deviations from the original data into two divisions and selects from one division to perturb certain cells in the original dataset. One proportion z-test is adopted in detection to identify these intentional perturbations and determine the ownership of the dataset. To detect the watermark in datasets where some tuple sequences are inconsistent with the original dataset, we propose a matching algorithm to locate tuples containing key cells. Experimental results on real and synthetic datasets validate the effectiveness of \ourmethod and demonstrate its applicability to various ML datasets and models.

There are several related and practical directions worthy of future research. 
While \ourmethod provides a general minimal effect in the performance of various ML models trained on tabular datasets, minimizing the performance effect for a specific ML model remains a pivotal challenge.
Meanwhile, considering more possible attack methods, supporting more types of datasets, and enlarging the scope of the threat model can be further studied.


\newpage
\balance
\bibliographystyle{ACM-Reference-Format}
\bibliography{watermark}

\end{document}